\documentclass[a4paper]{amsart}

\usepackage[utf8]{inputenc}

\usepackage{epsfig}
\usepackage{amssymb,amsthm,amsxtra}
\usepackage{enumerate}
\usepackage{verbatim}
\usepackage{graphicx}
\usepackage{mathrsfs} 
\usepackage{subfigure}
\usepackage[hidelinks]{hyperref}

\usepackage[goodsyntax]{virginialake}\xyoption{arc}

\usepackage[scr=rsfso]{mathalfa}

\theoremstyle{plain}
\newtheorem  {theorem}             {Theorem}    [section]
\newtheorem  {lemma}      [theorem]{Lemma}
\newtheorem  {proposition}[theorem]{Proposition}
\newtheorem  {corollary}  [theorem]{Corollary}

\theoremstyle{definition}
\newtheorem  {definition} [theorem]{Definition}

\newtheorem  {notation}   [theorem]{Notation}

\theoremstyle{remark}
\newtheorem  {remark}     [theorem]{Remark}
\newtheorem  {example}    [theorem]{Example}

\begin{document}

\title{Subatomic Proof Systems: Splittable Systems}

\author{Andrea Aler Tubella}
\address{IRIF, CNRS and Univ. Paris Diderot}

\author{Alessio Guglielmi}
\address{University of Bath}

\thanks{This research has been supported by EPSRC Project EP/K018868/1 \emph{Efficient and Natural Proof Systems} and by ANR Project ANR-15-CE25-0014  \emph{FISP}}

\begin{abstract}

This paper presents the first in a series of results that allow us to develop a theory providing finer control over the complexity of normalisation, and in particular of cut elimination. By considering atoms as self-dual non-commutative connectives, we are able to classify a vast class of inference rules in a uniform and very simple way. This allows us to define simple conditions that are easily verifiable and that ensure normalisation and cut elimination by way of a general theorem. In this paper we define and consider \emph{splittable systems}, which essentially comprise a large class of linear logics, including Multiplicative Linear Logic and BV, and we prove for them a \emph{splitting theorem}, guaranteeing cut elimination and other admissibility results as corollaries. In papers to follow, we will extend this result to non-linear logics. The final outcome will be a comprehensive theory giving a uniform treatment for most existing logics and providing a blueprint for the design of future proof systems.

\end{abstract}

\maketitle


\newdimen\vlsewidth
\newcommand\vltechoice{
   \mathchoice{\settowidth\vlsewidth{$\vlse$}
               \hbox to\vlsewidth{\hss$\varotimes$\hss}}
              {\settowidth\vlsewidth{$\vlse$}
               \hbox to\vlsewidth{\hss$\varotimes$\hss}}
              {\settowidth\vlsewidth{$\scriptstyle\vlse$}
               \hbox to\vlsewidth{\hss$\scriptstyle\varotimes$\hss}}
              {\settowidth\vlsewidth{$\scriptscriptstyle\vlse$}
               \hbox to\vlsewidth{\hss$\scriptscriptstyle\varotimes$\hss}}}
\newcommand\vlpachoice{
   \mathchoice{\settowidth\vlsewidth{$\vlse$}
               \hbox to\vlsewidth{\hss$\bindnasrepma$\hss}}
              {\settowidth\vlsewidth{$\vlse$}
               \hbox to\vlsewidth{\hss$\bindnasrepma$\hss}}
              {\settowidth\vlsewidth{$\scriptstyle\vlse$}
               \hbox to\vlsewidth{\hss$\scriptstyle\bindnasrepma$\hss}}
              {\settowidth\vlsewidth{$\scriptscriptstyle\vlse$}
               \hbox to\vlsewidth{\hss$\scriptscriptstyle\bindnasrepma$\hss}}}
\renewcommand\vlte {\mathbin{\vltechoice}}
\renewcommand\vlpa {\mathbin{\vlpachoice}}
  \newcommand\SA  {\mathsf{SA}}
  \newcommand\Sy  {\mathsf{S}}
  \newcommand\Uts  {\mskip.5mu\mathscr U\mskip-1mu}
  \newcommand\Rts  {\mskip.5mu\mathscr R\mskip-1mu}
  \newcommand\For  {\mskip 1mu\mathscr F\mskip-1mu}
  \newcommand\Pde  {\mskip 1mu\mathscr P\mskip-1mu}
  \newcommand\Ats  {\mskip.5mu\mathscr A\mskip-1mu}
 \newcommand\Gor  {\mskip 1mu\mathscr G\mskip-1mu}
  \newcommand\pr   {\mathop{\mathsf{pr}}}
  \newcommand\cn   {\mathop{\mathsf{cn}}}
  \newcommand\ideq {\mathrel\equiv}
  \newcommand\eq   {\mathrel=}
  \newcommand\nideq{\not\ideq}
\newcommand\rel {\mathrel\nu}  
\newcommand\relo {\mathrel{\overline{\nu}}}
\newcommand \relc {\mathrel {\color{blue} \nu}}
  \newcommand\rels { \mathrel\alpha}
    \newcommand\relso { \mathrel{\overline{\alpha}}}
\newcommand\relt {\mathrel\beta}
\newcommand\relto {\mathrel{\overline{\beta}}}  
\newcommand\relf {\mathrel\gamma} 
\newcommand\relsc {\mathrel {\color{red} \alpha}}
\newcommand\SAKS{\mathsf{SAKS}}
\newcommand\SALL {\mathsf{SALL}}
\newcommand\SALLS {\mathsf{SALLS}}
\newcommand\SAMALLS {\mathsf{SAMALLS}}
\newcommand\SMALLS {\mathsf{SMALLS}}
\newcommand\SABV {\mathsf{SABV}}
\newcommand\SABVU {\mathsf{SABVU}}
\newcommand\BVU {\mathsf{BVU}}
 \newcommand\SAMLL  {\mathsf{SAMLL}}
 \newcommand\SMLLS  {\mathsf{SMLLS}}
 \newcommand\SAMLLS  {\mathsf{SAMLLS}}
\newcommand\SLLS {\mathsf{SLLS}}
\newcommand\as{\mathbin{\mathbf{a}}}
\newcommand\asc{\mathbin{\color{red}{\textsf{a}}}}
\newcommand\vlorc{\mathbin{\color{blue}{\vlor}}}
\newcommand\vltec{\mathbin{\color{red}{\vlte}}}
\newcommand\bs{\mathbin{\textsf{b}}}
\newcommand\unit{1}
\newcommand\maxr{\times}
\newcommand\minr{+}
\newcommand\lengthpa[1]{\mathopen|#1\mathclose|_{\vlpa}}
\newcommand\htpa[1]{\mathopen|#1\mathclose|_{\vlpa}}
 \newcommand\Forcl  {\mskip 1mu\mathscr{F}_{cl}}
 \newcommand\Forll  {\mskip 1mu\mathscr{F}_{ll}}
 \newcommand\Forbv  {\mskip 1mu\mathscr{F}_{bv}}

  \newcommand\mos  {\mathord\star}
  \newcommand\mo   {\mathord\star\mskip-.4\thinmuskip\vlsbr}

\section{Introduction}

In traditional cut-elimination procedures in Gentzen theory, we eliminate cut instances from proofs by moving upwards instances of the \emph{mix} rule \cite{Gent:69:Investig:xi,Gall:93:Construc:mb}:
\[
\odn{\vdash mA, \Gamma \quad \vdash n \bar{A}, \Delta}{}{\vdash \Gamma, \Delta}{} \quad.
\]

This rule conflates one instance of cut and several instances of contraction. In the absence of contraction,  \emph{i.e.}, when $m = n = 1$, the procedure leaves the size of the proofs essentially the same; this happens, for example, in multiplicative linear logic. On the other hand, if contraction is indeed present, therefore $m > 1$ or $n > 1$, the size of the proofs can grow exponentially or more (depending on the logic and the presence of quantifiers). In other words, there are two mechanisms that have different effects on complexity: while cut elimination per se would keep proofs in the same complexity class (modulo polynomials), the presence of contraction makes for a jump to a higher complexity class. When a cut is moved up, all the contractions that it intercepts have to go up as well. This phenomenon is universal in Gentzen theory: what happens in the sequent calculus equally takes place in natural deduction and its computational interpretations. For example, the size of $\lambda$-terms shrinks under linear $\beta$-reduction. It should be interesting, and especially so for computational interpretations, to be capable of separating cut and contraction and normalise on each of them separately and in a natural way.
\vspace{2pt}

Can we do better than we do in Gentzen theory? More specifically, can we answer the following two questions?

\begin{enumerate}
\item Can we have more control on the complexity of cut elimination? An example of inadequate control in Gentzen is the inability to type optimal reductions in the $\lambda$-calculus. This is essentially due to the coarseness of $\beta$-reduction, that corresponds to the inability of Gentzen proof systems to deal with contraction in its atomic form.

\item Why does cut elimination work? One indication of the inadequate understanding provided by Gentzen theory is that each new proof system requires an ad-hoc cut elimination proof. One wonders whether some common structure exists, behind the plethora of proof systems and normalisation proofs, that would allow us to understand normalisation in a simpler and more effective way.
\end{enumerate}

In this paper we present the first in a series of results that try to answer those two questions. We give a unified characterisation of a general notion of normalisation that includes cut elimination, that works for a vast class of logics (including classical and linear logic) and that separates the normalisation procedures associated to cut and contraction. One crucial observation is that the two apparently unrelated inference mechanisms of contraction and cut are actually two manifestations of the same underlying inference shape, despite their having a wildly different behaviour from the point of view of complexity. The common shape does more than justifying cut and contraction: it is common to the vast majority of inference rules, across all proof systems for all logics. It is, to us, an astounding and unexpected finding, for which we do not have a convincing a priori justification at this point. For now, we can only say that it works and it has already helped the development of new proof systems.

The common shape is the following inference figure (which we present with a slight simplification that we can ignore for now):
\[
\odn{(A \rels B) \relt (C \rels D)}{}{(A \relt C) \rels (B \relt D)}{} \quad,
\]
where $A, B, C$ and $D$ are formulae and where ${\rels}$ and ${\relt}$ are connectives. The connectives ${\rels}$ and ${\relt}$ could be, respectively, conjunction and disjunction in some logic (not necessarily classical), but they also could be parallel and sequential composition in a process algebra, and in that case we would stipulate that ${\relt}$ is self-dual and non-commutative, as in pomset logic or in BV \cite{Reto:97:Pomset-L:fd,Gugl:06:A-System:kl}. ${\rels}$ and ${\relt}$ can also be something entirely different: this paper gives some conditions under which we are able to automatically provide normalisation procedures that, among other consequences, entail the admissibility of cut and other inference rules.

A contraction is obtained by the inference shape above by stipulating that an atom is the superposition of its two truth values, that are maintained coherent in a proof by requiring that the binary connective $\as$, representing the atom $a$, is non-commutative; this way the two superposed states are not mixed. Under this interpretation, we could read $(\fff \as \ttt)$ as $a$ and $(\ttt \as \fff)$ as its negation $\bar a$ (therefore $\as$ is also self-dual), while $(\fff \as \fff)$ and $(\ttt \as \ttt)$ would be interpreted as $\fff$ and $\ttt$, respectively. Therefore, we obtain contraction on atoms as follows:
\vspace{3pt}

\centerline{we read
$ \odn{(\fff \as \ttt) \vlor (\fff \as \ttt)}{}{(\fff \vlor \fff) \as (\ttt \vlor \ttt)}{}$ as  $\odn{a \vlor a}{}{a}{}$ and we read 
$\odn{(\ttt \as \fff) \vlor (\ttt \as \fff)}{}{(\ttt \vlor \ttt) \as (\fff \vlor \fff)}{}$ as $ \odn{\bar{a} \vlor \bar{a}}{}{\bar{a}}{}$.}
\vspace{3pt}

Under the same interpretation, we obtain an atomic version of the cut rule as follows:
\vspace{3pt}

\centerline{we read
$ \odn{(\fff \as \ttt) \vlan (\ttt \as \fff)}{}{(\fff \vlan \ttt) \as (\ttt \vlan \fff)}{}$ as  $\odn{a \vlan \bar a}{}{\fff}{}$ and we read 
$\odn{(\ttt \as \fff) \vlan (\fff \as \ttt)}{}{(\ttt \vlan \fff) \as (\fff \vlan \ttt)}{}$ as $ \odn{\bar{a} \vlan a}{}{\fff}{}$.}
\vspace{3pt}

Crucially, in order to define and operate on this new structure, we need to abandon Gentzen formalisms and use deep inference. One sufficient reason to do so is that we know that there cannot be analytic and complete Gentzen proof systems for a linear logic containing a self-dual, non-commutative connective, such as BV \cite{Tiu:06:A-System:ai} or any linear logic where the atoms would be regarded as self-dual non-commutative connectives.

The idea behind deep inference is to set proof composition free from the classical logic bias. While in Gentzen theory proofs are composed by a fixed mechanism that essentially maps the language constructors to a classical metalevel, in deep inference the metalevel is abolished: proofs are composed by the same connectives by which formulae are composed. In other words, if a language models a certain algebra, the proofs of the tautologies in that same language have the same algebraic structure as the language, enriched by an inference relation. The only requirement for inference is essentially its supporting a deduction theorem. A succinct survey of deep inference is \cite{Gugl:14:Deep-Inf:fj}, and up-do-date information about the literature is on the website \cite{Gugl::Deep-Inf:uq}. One reason that makes deep inference promising in our quest to answer the two questions above is that its ability to reduce contractions to their atomic form already led to the design of a typed $\lambda$-calculus that is much closer to optimality than what can be typed in Gentzen theory \cite{GundHeijPari:13:Atomic-L:fk} as well as provided a framework for an enhanced type system for interaction nets \cite{GimeMose:13:The-Stru:fk}. We use the deep-inference formalism called \emph{open deduction}. The way we define open deduction here provides for a more concrete alternative to \cite{GuglGundPari::A-Proof-:fk}, which is where the formalism was firstly defined.

Deep inference is a generalisation of Gentzen theory, but its features follow the tradition. In the deep inference literature, where many normalisation procedures have been designed for proof systems for several logics (classical, modal, linear, substructural), the separation of the two mechanisms is already apparent: we call the normalisation of contractions \emph{decomposition} and the proof of admissibility of cuts (and several other rules) \emph{splitting}. While the former procedure generates significant complexity, the latter does not. We answer the two questions above by establishing that the decomposition + splitting approach to normalisation is sufficiently general, and we find out what conditions on proof systems enable it. In other words, we characterise the proof systems that enjoy decomposition and splitting, so that their normalisation theory immediately follows. A benefit that we are already exploiting in current research is that this result helps in the design of new proof systems, because it reveals properties, leading to good normalisation procedures, that would otherwise be undetectable.

This paper is devoted to proving a splitting theorem for a large class of proof systems whose only inference shape is the one mentioned above, under certain conditions. The central notion of this paper is that of \emph{splittable system}. This, in turn, depends on a classification of inference rules of the above shape into two classes, those that behave like contraction and those that behave like cut. Splittable systems are those without contraction-like rules, therefore we can say that a splitting theorem is a generalisation of cut elimination for a generalisation of linear logics. This is what the present paper is about.
The \emph{subatomic} methodology presented in this paper has yielded further results, complementary to the ones presented here. Further publications will address contraction-like rules and their normalisation theorems.

The idea behind splitting is very simple, and is rooted in deep inference methods. It focuses on understanding the behaviour of the context around the cut, and in particular it consists on breaking down a proof in different pieces by following the logical connectives involved in the cut to find their duals. We show that we can then rearrange the different components of the proof to obtain a cut-free proof. Cut-elimination via splitting \cite{Gugl:06:A-System:kl} has been achieved in the deep inference systems for linear logic \cite{Kahr:08:Interact:ad, Stra:02:A-Local-:ul}, multiplicative exponential linear logic \cite{Stra:03:MELL-in-:oy}, the mixed commutative/non-commutative logic BV \cite{Gugl:06:A-System:kl} and its extension with linear exponentials NEL \cite{GuglStra:02:A-Non-co:dq}, and classical predicate logic  \cite{Brun:06:Cut-Elim:cq}.

This type of argument has been used to prove the admissibility of rules other than the atomic cut \cite{Gugl:06:A-System:kl}, showing that it can be applied to any connective that we can follow upwards in a proof. Thus, the splitting procedure hinges strongly on the dualities present in propositional logical systems (to find the duals of $a$ and $\bar a$) and on the regularity of deep inference rules (to follow the atoms in a proof). In this paper, we show that in systems where the scope of connectives only increases reading from bottom to top, \emph{i.e.}, in the aforementioned splittable systems, we can follow these connectives all the way up through the proof. Interestingly, the class of rules shown admissible is precisely the class of rules that allow us to make the cut atomic in deep inference formalisms. 

Achieving this simple characterisation of splittable systems gives us a full understanding of how the splitting procedure works, and why it has been used with success to prove the admissibility of different rules in several systems. We note that splitting is a global procedure: we need to study the proof as a whole to obtain a cut-free proof through splitting. Furthermore, splitting does not create meaningful complexity: the size of the cut-free proofs obtained by general splitting is linear on the size of the proofs with cut they come from, and splitting is a procedure of polynomial-time complexity. This is an interesting observation for the further study of complexity, since deep inference proofs are as long or shorter than sequent proofs \cite{BrusGugl:07:On-the-P:fk}.

\section{Subatomic systems}

Since our main aim is to study the interactions between rules, we will do so in the setting of deep inference \cite{Gugl::Deep-Inf:uq} where rules can be reduced to their atomic form providing great regularity in the inference rule schemes. In deep inference, derivations can be composed by the logical connectives that are used to compose formulae \cite{GuglGundPari::A-Proof-:fk}. For example, if
\[
\phi=\odv{A}{}{B}{}
\quad\mbox{and}\quad
\psi=\odv{C}{}{D}{}
\]
are two derivations in propositional logic, 
\[ \phi \land \psi=\odv{A}{}{B}{} \land \odv{C}{}{D}{} \quad \mbox{and} \quad
\phi \vlor \psi= \odv{A}{}{B}{} \vlor \odv{C}{}{D}{}\]
 are two valid derivations with premisses $A \vlan C$ and $A \vlor C$ and conclusions $B \vlan D$ and $B \vlor D$ respectively. In deep inference, rules can be applied at any depth inside a formula and as a result every contraction and cut instance can be locally transformed into their atomic variants by a local procedure of polynomial-size complexity \cite{BrunTiu:01:A-Local-:mz}.

This provides a surprising regularity in the inference rule schemes: it can be observed that in most deep inference systems all rules besides the atomic ones can be expressed as 
\[
\odn{(A \rels B) \relt (C \rels' D)}{}{(A \relt C) \rels (B \relt' D)}{} \quad \mbox{,}
\]
where $A, B,C,D$ are formulae and ${\rels}, {\relt}, {\rels}', {\relt}'$ are connectives. We call this rule shape a \emph{medial shape}.  Following this discovery, we will achieve an even greater regularity on the inference rules by looking even further, \emph{inside the atoms}. In this section we will introduce a new methodology through which we are able to represent \emph{every rule} as an instance of a \emph{single} inference rule scheme. This characterisation is not trivial: it is a delicate trade-off to impose restrictions on the possible assignments for  $\alpha, \beta,\gamma, \epsilon,\zeta,\eta$ that allow us to characterise systems that enjoy cut-elimination, but that are general enough to encompass the expressivity of a wide variety of logics. Indeed, the finding of these restrictions is the product of a long trial-and-error phase to obtain the desired generality together with the desired properties.

The main idea of this work is to consider atoms as self-dual, noncommutative binary connectives and to build formulae by freely composing units by atoms and the other connectives. We will consider the occurrences of an atom $\as$ as interpretations of more primitive expressions involving a noncommutative binary connective, denoted by $\as$. Two formulae $A$ and $B$ in the relation $\as$, in this order, are denoted by $A \as B$.  Formulae are built over the units for the connectives, denoted for example by $\ttt, \fff$ in the case of classical logic. We can think of it as a superposition of truth values: $\fff \as \ttt$ is the superposition of the two possible assignments for the atom $a$. We can for example have a projection onto a specific assignment by choosing which `side' we read: if we read the values on the left of the atom we assign $\fff$ to $a$ and if we read the ones on the right we assign $\ttt$ to $a$. We call these formulae \emph{subatomic}. For example, 
\[ ((\ttt \as \fff) \vlan (\fff \bs \fff)) \vlor ((\fff \vlan \ttt) \as \ttt) \quad \mbox{and} \quad (\ttt \as \ttt) \bs (\fff \vlan \fff)
\]
 are subatomic formulae for classical logic. 

In this way, we obtain an extended language of formulae which we will relate to the usual propositional formulae, or \emph{interpret}, through an interpretation map $\overset{I}{\mapsto}$ in Section \ref{int}.

By exploiting this new methodology, we will show through examples that we are strikingly able to present proof systems in such a way that \emph{every rule} has a medial shape, \emph{including the atomic rules} that do not usually follow this scheme. 


\subsection{Subatomic formulae}

Subatomic formulae are built by freely composing constants by connectives and atoms. 
By considering atoms as connectives we will work with an extended language of formulae, since we can have atoms in the scope of other atoms, something that does not occur in `traditional' formulae.

\begin{definition}
Let $\Uts$ be a denumerable set of \emph{constants} whose elements are denoted by $u, v, w, \dots$. Let $\Rts$ be a denumerable partially ordered set of \emph{connectives} whose elements are denoted by $\rels$, $\relt$, $\relf$, \dots. The set $\For$ of \emph{subatomic formulae} (or $\SA$ formulae) contains terms defined by the grammar
\[
\For\mathrel{::=}\Uts\mid\For \; \Rts \; \For
\quad\mbox{.}
\]
Formulae are denoted by $A$, $B$, $C$, \dots.

\vspace{2pt}
 A \emph{(formula) context} $K\vlhole\cdots\vlhole$ is a formula where some subformulae are substituted by holes; $K\{A_1\}\cdots\{A_n\}$ denotes a formula where the $n$ holes in $K\vlhole\cdots\vlhole$ have been filled with $A_1$, \dots, $A_n$.

\vspace{2pt}
 The expression $A\ideq B$ means that the formulae $A$ and $B$ are syntactically equal. We omit parentheses when there is no ambiguity.
\vspace{2pt}

In $K\{A \rels B\}$ we say that the subformulae of $A$ and $B$ are \emph{in the scope of} $\rels$.
\end{definition}

We show the subatomic methodology for classical logic and for multiplicative linear logic.

\begin{example}\label{exCL}
The set $\Forcl$ of subatomic formulae for classical logic is given by the set of constants $\Uts=\{\fff, \ttt\}$ and the set of connectives $\Rts=\{\vlan, \vlor\}\cup \Ats$ where $\Ats$ is a denumerable set of atoms, denoted by ${\as} , {\bs}, \dots$ with $\Ats \cap \{\vlan, \vlor\}= \varnothing$. Two examples of subatomic formulae for classical logic are
$$A\ideq ((\fff \as \ttt) \vlor (\ttt \as \ttt)) \vlan (\ttt \bs \fff) \quad \mbox{and} \quad B\ideq ((\ttt \bs \fff) \vlan \ttt ) \vlor (\fff \as \fff) \quad \mbox{.}$$
\end{example}

\begin{example}\label{exMLL}
\vllineartrue
The set $\Forll$ of subatomic formulae for multiplicative linear logic is given by the set of constants $\Uts=\{\bot,  \one \}$ and the set of connectives $\Rts=\{\vlpa, \vlte\} \cup \Ats$ where $\Ats$ is a denumerable set of atoms, denoted by ${\as}, {\bs}, \dots$ with  $\Ats \cap \{\vlpa, \vlte\}= \varnothing$. Two examples of subatomic formulae for linear logic are 
$$C\ideq ((\one \vlpa \bot) \as \one)\vlte \bot \quad \mbox{and} \quad D\ideq ((\bot \vlpa \one) \bs \one) \vlte (\one \as \bot) \quad \mbox{.}$$
\vllinearfalse
\end{example}

Aside from classical logic and multiplicative linear logic, we will feature the logic $\BV$ \cite{Gugl:06:A-System:kl, Kahr:07:System-B:fk, Kahr:09:On-Linea:ix} amongst the examples to showcase a well-studied logic with self-dual non-commutative connectives. For that, we define the logic $\BVU$. $\BV$ will correspond to $\BVU$ with all the units identified.

\vspace{6pt}

\begin{example}\label{exBV}
We define system $\BVU$. The formulae of $\BVU$ are built from the units  $\bot, \circ, \one$ by composing them with the connectives $\vlpa$, $\triangleleft$, $\vlte$. 

The connectives $\vlpa$ and $\vlte$ are dual to each other, associative, commutative and have units $\bot$ and $\one$ respectively. $\triangleleft$ is self-dual and associative, and has unit $\circ$. 

Negation on $\BVU$ formulae is built respecting De Morgan dualities, with $\bar \circ =\circ$ and $\bar \bot=\one$.

The units verify the equations $\circ \vlpa \circ=\one$, $\circ \vlte \circ=\bot$ and $\one \triangleleft \one=\one$, $\bot \triangleleft \bot=\bot$.

The inference rules for system $\BVU$ are given by the same rules as for system $\BV$ \cite{Gugl:06:A-System:kl}. System $\BV$ corresponds to system $\BV$ with the three units identified,  \emph{i.e.}, $\one=\circ=\bot$.
\vspace{6pt}

The set $\Forbv$ of subatomic formulae for the non-commutative logics $\BVU$ and $\BV$ is given by the set of constants $\Uts=\{ \bot, \one, \circ \}$ and the set of connectives $\Rts=\{{\vlpa}, {\triangleleft}, {\vlte}\} \cup \Ats$ where $\Ats$ is a denumerable set of atoms, denoted by $a, b, \dots$ with  $\Ats \cap \{{\vlpa}, {\triangleleft}, {\vlte}\}= \varnothing$. Two examples of subatomic formulae for BV are
$$E\ideq (\one \as \bot) \triangleleft (\circ \vlte (\bot \bs \bot)) \quad \mbox{and} \quad F \ideq ((\circ \vlte \one) \as \one) \vlpa \one \quad \mbox{.}$$ 
\end{example}

Just like for `ordinary' formulae, we will define an equational theory and a negation map on the set of subatomic formulae. We will work in a classical setting, in the sense that we will consider an involutive negation that satisfies De Morgan dualities.

Furthermore, in order to keep track of the equational theory in the general cut-elimination result, we specify the equalities that we allow. 
Since we consider atoms as connectives, we will allow for a broad range of behaviours for the connectives, not assuming any logical characteristics or properties such as commutativity or associativity. We will therefore encompass logics with both commutative and non-commutative, associative and non-associative, dual and-self dual connectives. This feature deserves to be highlighted since expressing self-dual non-commutative connectives into proof systems that enjoy cut-elimination is a challenge in Gentzen-style sequent calculi: it is impossible to have a complete analytic system with a self-dual non-commutative connective \cite{Tiu:06:A-System:ai}.

\begin{definition}\label{eq}
We define \emph{negation} as a pair of involutive maps $\bar{\cdot}: \Rts \mapsto \Rts$ and  $\bar{\cdot}: \Uts \mapsto \Uts$. We define the \emph{negation map on formulae} as the map inductively defined by setting $\overline{A \rels B} := \overline{A} \relso \overline{B}$.
\vspace{4pt}

We define an equational theory $=$ on $\For$ as the minimal equivalence relation closed under negation (if $A=B$, then $\bar A= \bar B$) and under context (if $A=B$, then $K\{A\}=K\{B\}$ for any context $K\vlhole$) defined by any subset of axioms of the form:
\[
\begin{array}{@{}c@{\quad}l@{%
\kern-35pt
\qquad}r@{}}
(1)&\mbox{$\forall A, B, C \in \For$. $(A \rels B) \rels C=A \rels (B \rels C)$ };
&\emph{(Associativity of  $\rels$)}\\
(2)&\mbox{$\forall A, B\in \For$. $A \rels B= B \rels A$ };
&\emph{(Commutativity of $\rels$)}\\
(3)&\mbox{$\forall A \in \For$.  $A \rels u_{\rels}=A=u_{\rels} \rels A$ for a fixed $u_{\rels} \in \Uts$ };
&\emph{(Unit of $\rels$)}\\
(4)&\mbox{$v \rels w=u$  for fixed $v,w,u\in \Uts$ };
&\emph{(Constant assignment for $\rels$)}\\
(5)&\mbox{$u=v$ for fixed $u,v\in\Uts$ }.
&\emph{(Constant identification)}\\
\end{array}
\]

If there is an axiom of the form (1) for $\rels$, we say that $\rels$ is \emph{associative}. If there is an axiom of the form (2) for $\rels$, we say that $\rels$ is \emph{commutative}. If there is an axiom of the form (3) for $\rels$ we say that $\rels$ is \emph{unitary}, and we call $u_{\rels}$ the \emph{unit of} $\rels$.

\end{definition}

\begin{remark}
Since the equational theory is closed under negation, if $\rels$ is unitary with unit $u_{\rels}$, then $\relso$ is unitary and its unit is $\overline{ u_{\rels}}$.
\end{remark}

\begin{example}\label{exCLeq}
For the set of subatomic formulae for classical logic $\Forcl$ defined in example \ref{exCL}, we define negation through:
\[
\begin{array}{l}
\bar {\vlan} \ideq {\vlor} \mbox{ ;}\\
\bar {\as} \ideq {\as}  \mbox{ for all } {\as} \in \Ats  \mbox{ ;}\\
\bar \ttt \ideq \fff \mbox{ .}\\
\end{array}
\]

We define the equational theory $=$ on $\Forcl$ as the minimal equivalence relation closed under negation and under context defined by:
\[
\begin{array}{ll}
\mbox{For all } A, B, C \in \For:&\\[8pt]
 (A \vlan B) \vlan C= A \vlan (B \vlan C)  \mbox{ ;} &\qquad  (A \vlor B) \vlor C= A \vlan (B \vlor C)  \mbox{ ;}\\
 A \vlan B=B \vlan A \mbox{ ;} &\qquad A \vlor B=B \vlor A \mbox{ ;}\\[6pt]
A \vlan \ttt=A \mbox{ ;} & \qquad A \vlor \fff = A  \mbox{ ;} \\
 \fff \vlan \fff=\fff \mbox{ ;} & \qquad \ttt \vlor \ttt =\ttt \mbox{ ;}\\[6pt]
 \forall {\as} \in \Ats .  \fff \as \fff=\fff \mbox{ ;} &\qquad  \forall {\as} \in \Ats . \ttt \as \ttt =\ttt \mbox{ .}\\
\end{array}
\]
\end{example}

\begin{example}\label{exMLLeq}
\vllineartrue
For the set of subatomic formulae for linear logic $\Forll$ defined in example \ref{exMLL}, we define negation through:
\[
\begin{array}{l}
\bar {\vlte} \ideq {\vlpa}\mbox{ ;}\\
\bar {\as} \ideq {\as}  \mbox{ for all } {\as} \in \Ats \mbox{ ;}\\
\bar \one \ideq \bot \mbox{ .}\\
\end{array}
\]

We define the equational theory $=$ on $\Forll$ as the minimal equivalence relation closed under negation and under context defined by:
\[
\begin{array}{ll}
\mbox{For all } A, B, C \in \For:&\\[8pt]
 (A \vlte B) \vlte C= A \vlte (B \vlte C) \mbox{ ;} &\qquad (A \vlpa B) \vlpa C= A \vlpa (B \vlpa C) \mbox{ ;}\\
 A \vlte B=B \vlte A \mbox{ ;} &\qquad  A \vlpa B=B \vlpa A \mbox{ ;}\\[6pt]
 A \vlte \one=A \mbox{ ;} & \qquad A \vlpa \bot = A  \mbox{ ;}\\[6pt]
 \forall {\as} \in \Ats .  \bot \as \bot=\bot \mbox{ ;} &\qquad  \forall {\as} \in \Ats .  \one \as \one =\one \mbox{ .}\\
\end{array}
\]
\vllinearfalse
\end{example}

\begin{example}\label{exBVeq}
For both $\BVU$ and $\BV$ we will define the same negation map. They will differ only on the equational theory, since all the units are identified in $\BV$.

For the set of subatomic formulae for $\BVU$ and for $\BV$ $\Forbv$ defined in example \ref{exBV}, we define negation through:
\[
\begin{array}{l}
\bar {\vlte} \ideq \vlpa \mbox{ ;}\\
\bar \triangleleft \ideq \triangleleft \mbox{ ;}\\
\bar a \ideq a \ \mbox{for all } a \in \Ats \mbox{ ;}\\
\bar \circ \ideq \circ \mbox{ ;} \\
\bar \bot \ideq \one \mbox{ .}\\
\end{array}
\]

For the logic $\BVU$ we define an equational theory $=$ on $\Forbv$ as the minimal equivalence relation closed under negation and under context defined by:
\[
\begin{array}{ll}
\mbox{For all } A, B, C \in \For:&\\[8pt]
 (A \vlte B) \vlte C= A \vlte (B \vlte C) \mbox{ ;} &\qquad (A \vlpa B) \vlpa C= A \vlpa (B \vlpa C) \mbox{ ;}\\
A \vlte B=B \vlte A \mbox{ ;} &\qquad A \vlpa B=B \vlpa A \mbox{ ;}\\
(A \triangleleft B) \triangleleft C =A \triangleleft (B \triangleleft C)\mbox{ ;} & \\[6pt]
 A \vlte \one=A \mbox{ ;} & \qquad A \vlpa \bot = A \mbox{ ;} \\
A \triangleleft \circ = A \mbox{ ;} &\qquad \circ \triangleleft A=A \mbox{ ;} \\[6pt]
\circ \vlte \circ = \bot \mbox{ ;} & \qquad \circ \vlpa \circ =\one \mbox{ ;} \\[6pt]
\forall a \in \Ats .  \bot \as \bot=\bot \mbox{ ;} &\qquad  \forall a \in \Ats.  \one \as \one =\one \mbox{ ;}\\
 \bot \triangleleft \bot=\bot \mbox{ ;} &\qquad \one \triangleleft \one =\one \mbox{ .}\\

\end{array}
\]

The equational theory for the logic $\BV$ defined on the set of subatomic formulae $\Forbv$ is given by the previous equations, together with the added axioms:
\[
\begin{array}{l@{\qquad \qquad \qquad \qquad \qquad  }l@{\qquad\qquad\qquad\qquad}}
 \one = \circ \mbox{ ;}  &\bot = \circ \mbox{ .} \\
\end{array}
\]

\end{example}

\begin{remark}
To define $\BV$, it would be possible to replace the equations $\circ \vlte \circ = \bot $ and $\circ \vlpa \circ =\one $ by inference rules
\[
\odn{\one}{}{\circ \vlpa \circ}{} \quad \mbox{and} \quad \odn{\circ \vlte \circ}{}{\bot}{}\mbox{ .}
\]
Since these rules only operate on units rather than being a rule scheme, the uniformity provided by the subatomic approach is still very beneficial. The splitting results presented in this paper would still hold with this replacement, by modifying them very slightly to allow for rules involving units. 

For the sake of brevity and since in $\BV$ all the units are identified, in the examples presented in this paper we choose to work with the above equations rather than with inference rules. Nonetheless, the approach to representing $\BV$ may vary as research continues in this area.
\end{remark}

Given a propositional logic with certain connectives and constants, its subatomic counterpart is therefore composed of an extended language of formulae, made up from the same connectives but with the added possibility of having atoms in the scope of other atoms. 

\subsection{Subatomic proof systems}

The useful properties of subatomic formulae become apparent when we extend the principle to atomic inference rules. Let us consider, for example, the usual contraction
rule for an atom in classical logic given by
 \[
\odn{a \vlor a}{}{a}{} \quad.
\] 

We could obtain this rule subatomically by reading $\fff \as \ttt$ as $a$ and $\ttt \as \fff$ as $\bar a$, as follows:
\vspace{4pt}

\centerline{we read
$ \odn{(\fff \as \ttt) \vlor (\fff \as \ttt)}{}{(\fff \vlor \fff) \as (\ttt \vlor \ttt)}{}$ as  $\odn{a \vlor a}{}{a}{}$ and we read 
$\odn{(\ttt \as \fff) \vlor (\ttt \as \fff)}{}{(\ttt \vlor \ttt) \as (\ttt \vlor \ttt)}{}$ as $ \odn{\bar{a} \vlor \bar{a}}{}{\bar{a}}{}$.}
\vspace{4pt}

These rules are therefore generated by the linear scheme
\[
\odn{(A \as B) \vlor (C \as D)}{}{(A \vlor C) \as (B \vlor D)}{} \mbox{, where $A,B,C,D$ are formulae.}
\]

The non-linearity of the contraction rule has been pushed from the atoms to the units.

Similarly, we can consider the atomic identity rule 
\[\odn{\ttt}{}{a \vlor \bar{a}}{} \quad.\] 

It can be obtained subatomically as follows:
\vspace{4pt}

\centerline{we read
$ \odn{(\fff \as \ttt) \vlor (\ttt \as \fff)}{}{(\fff \vlor \ttt) \as (\ttt \vlor \fff)}{}$ as  $\odn{\ttt}{}{a \vlor \bar a}{}$.}
\vspace{4pt}

Similarly to the contraction rule,  it is generated by the linear scheme 
\[
\odn{( A \vlor B) \as (C \vlor D)}{}{(A \as C) \vlor (B \as D)}{} \mbox{, where $A,B,C,D$ are formulae.}
\]

It is quite plain to see that both the subatomic contraction rule and the subatomic introduction rule have the same shape. This regularity is made useful in combination with the observation that
in fact the linear rule scheme
\[
\odn{(A \rels B) \relt (C \rels' D)}{}{(A \relt C) \rels (B \relt' D)}{} \quad \mbox{,}
\]
where ${\rels}, {\relt}, {\rels}', {\relt}'$ are connectives, and $A,B,C,D$ are formulae is typical of logical rules in deep inference. We refer to it as a \emph{medial shape}. For example, consider system $\SKS$ for classical logic in Figure \ref{SKS}.

\begin{figure}
\centering
\fbox{
\begin{tabular}[c]{cc} 

\odn{ \ttt}{ai\downarrow}
{a \vlor \bar a}{}&

\odn{a \vlan \bar a }{ai\uparrow}
{ \fff}{}\\[15pt]

\odn{ (A \vlor B)\vlan C}{s}
{ (A\vlan C)\vlor B}{} &

\odn{ (A\vlan B)\vlor (C\vlan D)}{m}
{(A\vlor C)\vlan (B\vlor D)}{}\\[15pt]

\odn{a \vlor a}{ac \downarrow}
{ a}{}&

\odn{a}{ac \uparrow}
{ a \vlan a}{}\\[15pt]

\odn{\fff}{aw\downarrow}{a}{}&

\odn{a}{aw\uparrow}{\ttt}{}\\

\end{tabular}}

\vspace{3pt}
\caption{ System $\SKS$}
 \label{SKS}
\end{figure}

We can see that the rule $m$ follows this scheme as well, and
we can derive the rule $s$ from the rule
\[
\odn{ (A \vlor B)\vlan (C \vlor D)}{\vlan\downarrow}
{ (A\vlan C)\vlor (B \vlor D)}{} \quad,
\]
which follows this scheme. We have therefore uncovered an underlying structure behind the shape of inference rules, that we will exploit to obtain a general characterisation of rules.

To make use of the general characterisation, we will impose some restrictions on  ${\rels}, {\rel}, {\relt}, {\relf}$. These conditions strike a balance between being general enough to encompass a wide variety of logics and being explicit enough to enable us to generalise procedure such as cut-elimination and decomposition. 

The restrictions on the connectives of the rule scheme stem from the observation that certain dualities between the connectives are maintained in every rule. For example, we can write the rule $\vlan\!\!\downarrow$ as 
\[
\odn{ (A \vlor B)\vlan (C\vlor D)}{\vlan \downarrow}
{ (A\vlan C)\vlor (B \mathbin{\bar \vlan} D)}{} \quad
\]

and the atomic identity rule as 

\[
\odn{ (A \vlor B)\as (C\vlor D)}{}
{ (A \as C)\vlor (B \mathbin{\bar {\as}} D)}{} \quad \mbox{.}
\]

We will generalise this observation, considering rules with a medial shape and certain dualities between the connectives involved and show that this shape is enough to represent a wide variety of logics. With the subatomic methodology, we are therefore able to represent proof systems in such a way that every rule has the same shape. This full regularity gives us a newly gained ability to characterise proof systems that enjoy properties such as decomposition and cut-elimination.
\vspace{4pt}

To characterise the dualities present in the inference rules, we introduce a notion of polarity in the pairs of dual relations. This notion of polarity can be reminiscent of the polarities assigned to connectives in linear logic \cite{Gir:91:Anewcon}, but the idea behind it is rather to assign which of the relations in the pair is `stronger' than the other. Intuitively, it loosely corresponds to assigning which relation of the pair will imply the other. For example, in classical logic $A \vlan B$ implies $A \vlor B$, and thus we will assign $\vlan$ to be \emph{strong} and $\vlor$ to be \emph{weak}. 

\begin{definition}
For each pair of connectives $\{\rels, \relso\}$, we give a polarity assignment: we call one connective of the pair \emph{strong} and the other one \emph{weak}. 

If $\rels$ is strong and $\relso$ is weak, we will write $\rels^M=\relso^M=\rels$ and $\rels^m=\relso^m=\relso$. Self-dual connectives are both strong and weak.
\end{definition}

\newcommand\downrule{ $\odn{ ( A \relt B)  \rels  (C \relt D) }{}{(  A \rels  C) \relt (B \rels^m D )}{}$}
\newcommand\uprule{$\odn{(A \relt B) \rels  (C \relt^M D)}{}{(A \rels C) \relt (B \rels D)}{}$}

\begin{definition}
A \emph{subatomic proof system} $\SA$ with set of formulae $\For$ is
\begin{itemize}
\item a collection of inference rules of the shape \downrule, $\rels, \relt \in \Rts$, called \emph{down-rules},
\item a collection of inference rules of the shape \uprule, $\rels, \relt \in \Rts$, called \emph{up-rules}, 
\item a collection of rules $\odn A\eq B{}$ and $\odn {\overline{A}}{\eq}{\overline B}{}$, for every axiom $A=B$ of the equational theory $=$ on $\For$, called \emph{equality rules}.
\end{itemize}
\end{definition}
 
Note that the non-invertible rules are linear.

\begin{remark}
Since we will not always work modulo equality, we define the equality rules to be inference steps just like the inference rules, rather than focusing on equality as equations between formulae. Two formulae $A$ and $B$ will be equal if and only if there is a derivation from $A$ to $B$ composed only of equality rules.

We could have just as well defined equality between formulae directly in this way, but chose to define it initially as an equivalence relation for the sake of clearer exposition.

The rules $\odn A\eq B{}$ are invertible and correspond to equivalence by mutual implication. Every non-invertible rule with logical significance is therefore an instance of the general rule scheme with medial shape.
\end{remark}

\begin{remark}
We will often use the notation
\[
\odn{(A\relt B) \rels^M (C \relt D)}{}{(A \rels B) \relt (C \relso D)}{} 
\] 
for down-rules with a strong connective in the premiss where $\relt$ is commutative.
\end{remark}

\begin{example}
We consider $\vlan $ as strong and \/ $\vlor$ as weak in classical logic. The subatomic proof system $\SAKS$ is given by the inference rules in Figure \ref{SAKS}, together with the equality rules given by $\odn{A}{=}{B}{}$ for every $A$, $B$ on opposite sides of the equality axioms provided in example \ref{exCLeq}. 

Rules labelled with $\downarrow$ are down-rules, and rules labelled by \/ $\uparrow$ are up-rules. The medial rule labelled by $m$ is self-dual, and is both a down-rule and an up-rule.
\end{example}

\begin{figure}
\centering
\fbox{
\begin{tabular}{c} 

\odn{ (A\vlor B)\as  (C\vlor D)}{{\as}\downarrow}
{( A\as C ) \vlor ( B\as D)}{}\qquad \qquad

\odn{( A\as B) \vlan ( C\as D)}{{\as}\uparrow}
{ (A\vlan C)\as  (B\vlan D)}{}\\[15pt]

\odn{ (A \vlor B)\vlan (C\vlor D)}{\vlan \downarrow}
{ (A\vlan C)\vlor (B \vlor D)}{} \qquad \qquad

\odn{ (A \vlor B)\vlan (C\vlan D)}{\vlor \uparrow}
{ (A\vlan C)\vlor (B \vlan D)}{} \\[15pt]

\odn{ (A\vlan B)\vlor (C\vlan D)}{m}
{(A\vlor C)\vlan (B\vlor D)}{}\\[15pt]

\odn{( A \as B ) \vlor ( C \as D)}{{\as}c}
{ (A \vlor C) \as (B \vlor D)}{}\qquad \qquad 

\odn{( A \vlan B )\as ( C \vlan D )}{{\as}\bar c }
{ ( A \as C ) \vlan ( B \as D)}{}\\

\end{tabular}}
\caption{$\SAKS$}
\label{SAKS}

\end{figure}

\begin{example}\label{exSAMLLS}
We consider \/ $\vlte$ as strong and \/ $\vlpa$ as weak in multiplicative linear logic. The subatomic proof system $\SAMLLS$ is given by the inference rules in Figure \ref{SAMLLS} together with the equality rules given by $\odn{A}{=}{B}{}$ for every $A$, $B$ on opposite sides of the equality axioms provided in example \ref{exMLLeq}.  

\end{example}

\begin{figure}
\centering
\vllineartrue
\fbox{
\begin{tabular}{c} 

\odn{ (A \vlpa B) \as  (C \vlpa D)}{{\as}\downarrow}
{( A\as C) \vlpa ( B \as D)}{}\qquad \qquad

\odn{( A\as B ) \vlte ( C\as D)}{{\as}\uparrow}
{ (A \vlte C)\as  (B \vlte D)}{}\\[15pt]

\odn{ (A \vlpa B) \vlte  (C \vlpa D)}{\vlte\downarrow}
{( A\vlte C) \vlpa ( B \vlpa D)}{}\qquad \qquad

\odn{( A\vlpa B ) \vlte ( C\vlte D)}{\vlpa\uparrow}
{ (A \vlte C)\vlpa  (B \vlte D)}{}\\

\end{tabular}}
\caption{$\SAMLLS$}
\label{SAMLLS}
\end{figure}

\begin{figure}
\fbox{
\begin{tabular}{c} 

\odn{ \one}{ai\downarrow}
{a \vlpa \bar a}{}\qquad \qquad \qquad \qquad \qquad \quad

\odn{a \vlte \bar a}{ai\uparrow}
{\bot}{}\\[15pt]

\odn{( A\vlpa B ) \vlte C}{s}
{ (A \vlte C)\vlpa  B}{}\\[12pt]

\end{tabular}}
\caption{$\SMLLS$  \cite{Stra:03:Linear-L:lp}}
\label{SMLLS}
\vllinearfalse
\end{figure}

\begin{example}
We consider $\vlte$ as strong and $\vlpa$ as weak in $\BVU$ and $\BV$. The subatomic proof system $\SABVU$ is given by the inference rules in Figure \ref{SABV} together with the equality rules given by $\odn{A}{=}{B}{}$ for every $A$, $B$ on opposite sides of the equality axioms for $\BVU$ provided in example \ref{exBVeq}.

Likewise, the subatomic proof system $\SABV$ is given by the same inference rules and equality rules, together with the equality rules given by $\odn{\bot}{=}{\circ}{}$, $\odn{\one}{=}{\circ}{}$ and their converse.
\vllinearfalse
\end{example}

\begin{figure}
\vllineartrue
\centering
\fbox{
\begin{tabular}{c} 

\odn{ (A \vlpa B) \as  (C \vlpa D)}{\as\downarrow}
{( A\as C) \vlpa ( B \as D)}{}\qquad \qquad

\odn{( A\as B ) \vlte ( C\as D)}{\as\uparrow}
{ (A \vlte C)\as  (B \vlte D)}{}\\[15pt]

\odn{ (A \vlpa B) \vlte  (C \vlpa D)}{\vlte\downarrow}
{( A\vlte C) \vlpa ( B \vlpa D)}{}\qquad \qquad

\odn{( A\vlpa B ) \vlte ( C\vlte D)}{\vlpa\uparrow}
{ (A \vlte C)\vlpa  (B \vlte D)}{}\\[15pt]

\odn{ (A \vlpa B) \triangleleft  (C \vlpa D)}{\triangleleft \downarrow}
{( A\triangleleft C) \vlpa ( B \triangleleft D)}{}\qquad \qquad

\odn{( A\triangleleft B ) \vlte ( C\triangleleft D)}{\triangleleft \uparrow}
{ (A \vlte C)\triangleleft  (B \vlte D)}{}\\

\end{tabular}}
\caption{$\SABV$}
\label{SABV}

\vspace{8pt}
\fbox{
\begin{tabular}{c} 
\odn{ \circ}{ai\downarrow}
{a \vlpa \bar a}{}\qquad \qquad \qquad \qquad \qquad \quad

\odn{a \vlte \bar a}{ai\uparrow}
{ \circ}{}\\[15pt]

\odn{( A\vlpa B ) \vlte C)}{s}
{ (A \vlte C)\vlpa  B}{}\\[15pt]

\odn{ (A \vlpa B) \triangleleft  (C \vlpa D)}{q\downarrow}
{( A\triangleleft C) \vlpa ( B \triangleleft D)}{}\qquad \qquad

\odn{( A\triangleleft B ) \vlte ( C\triangleleft D)}{q\uparrow}
{ (A \vlte C)\triangleleft  (B \vlte D)}{}\\
\end{tabular}}
\caption{$\SBV$ \cite{Gugl:06:A-System:kl} }
\label{SBV}
\vllinearfalse
\end{figure}

To reduce rules to their subatomic form, we will work in the setting of deep inference \cite{Gugl::Deep-Inf:uq}, where derivations can be composed with the same connectives as formulae. The deep inference methodology has been exploited in many ways, such as shortening analytic proofs by exponential factors with respect to Gentzen proofs \cite{BrusGugl:07:On-the-P:fk, Das:11:On-the-P:fk}, modelling process algebras \cite{Brus:02:A-Purely:wd, Kahr:05:Towards-:uz, Rove:11:Linear-L:kx, Rove:12:Extendin:uq} or typing optimised versions of the $\lambda$-calculus that provide a novel treatment of sharing and duplication \cite{GundHeijPari:13:Atomic-L:fk}. The particular property that most interests us is that rules can be applied at any depth inside a formula and as a result every contraction and cut instances can be locally transformed into their atomic variants by a local procedure of polynomial-size complexity \cite{BrunTiu:01:A-Local-:mz}. Therefore they can be transformed into their subatomic variants straightforwardly. 

We will present derivations in the open deduction formalism \cite{GuglGundPari::A-Proof-:fk}, which is a logic-independent formalism, allowing us to reach the desired level of generality. In open deduction, derivations can be presented in their `synchronal' form, where there is maximal parallelism between inference steps, and in sequential form, where inference steps are in a total order, and in all the intermediate forms. There are simple (non-confluent) algorithms that non-deterministically transform a derivation (in particular a synchronal one) into one among several sequential derivations with the same premiss and conclusion. Conversely, simple (confluent) algorithms transform derivations (in particular sequential ones) into synchronal derivations.

\begin{definition}
Given a subatomic systems $\SA$ and formulae $A$ and $B$, a \emph{derivation} $\phi$ in $\SA$ from \emph{premiss} $A$ to \emph{conclusion} $B$ denoted by $\odbox{\odv{A}{\phi}{B}{\SA}}$ is defined to be:  

\begin{itemize}
\item a formula $\phi \ideq A \ideq B$;
\item a \emph{composition by inference} 
\[
\phi\ideq \odn{\odv{A}{\phi_1}{A'}{\SA}}{\rho}{\odv{B'}{\phi_2}{B}{\SA}}{}
\]
 where $\rho$ is an instance of an inference rule in $\SA$ and $\phi_1$ and $\phi_2$ are derivations in $\SA$;
\item a \emph{composition by connectives}
\[
\phi\ideq \odbox{ \odv{A_1}{\phi_1}{B_1}{\SA} }\rels \odbox{\odv{A_2}{\phi_2}{B_2}{\SA}}
\]
where $\rels \in \Rts$, $A \ideq A_1 \rels A_2$, $B\ideq B_1 \rels B_2$, $\phi_1$ and $\phi_2$ are derivations in $\SA$.
\end{itemize}

We denote by
\[
\odbox{\odv{A}{\phi}{B}{\{\rho_1, \dots, \rho_n\}}}
\]
a derivation where only the rules $\rho_1, \dots, \rho_n$ appear.

Sometimes we omit the name of a derivation or the name of the proof system if there is no ambiguity.

To improve readability sometimes we remove the boxes around derivations.
\end{definition}

\begin{notation}
We consider the two derivations
\[
\odframefalse
\odn{(
\odn{\odv{A_1}{\phi_1}{A_2}{\SA}}{\rho_1}{\odv{B_1}{\phi_2}{B_2}{\SA}}{}
)}{\rho_2}
{\odv{C_1}{\phi_3}{C_2}{\SA}}{} \qquad \mbox{and} \qquad
\odn{\odv{A_1}{\phi_1}{A_2}{\SA}}{\rho_1}
{(\odn{\odv{B_1}{\phi_2}{B_2}{\SA}}{\rho_2}
{\odv{C_1}{\phi_3}{C_2}{\SA}}{})}{}
\]

to be equal and we denote them both by 
\[
\od{\odd{\odi{\odd{\odi{\odd{
\odh{A_1}}{\phi_1}{A_2}{\SA}}{\rho_1}
{B_1}{}}{\phi_2}
{B_2}{\SA}}{\rho_2}
{C_1}{}}{\phi_2}
{C_2}{\SA}} \qquad \mbox{.}
\]

\end{notation}

\begin{example}\label{exderiv}
The following is a $\SAKS$ derivation with premiss $(\fff \vlor \ttt) \as (\ttt \vlor \fff) \vlan ( (\fff \bs \ttt) \vlor \ttt) \vlan \fff$ and conclusion $((\fff \as \ttt) \vlan (\fff \bs \ttt) ) \vlor ((\ttt \as \fff) \vlor \ttt) \vlan \fff$:
\[
\odbox{\odn{\odn{(\fff \vlor \ttt) \as (\ttt \vlor \fff) }{\as\downarrow}{(\fff \as \ttt) \vlor (\ttt \as \fff)}{}
\vlan
((\fff \bs \ttt) \vlor \ttt)}{s}
{((\fff \as \ttt) \vlan (\fff \bs \ttt) ) \vlor ((\ttt \as \fff) \vlor \ttt)}{}} \vlan \fff
\]
\end{example}

\begin{definition}
Let $\odv{A}{\phi}{B}{\SA}$ and $\odv{B}{\psi}{C}{\SA}$ be two derivations. We define their \emph{composition} $\odt\phi{}\psi{}$ as the derivation constructed as follows:
\begin{itemize}
\item[-] if $\phi$ is a formula then $\odt\phi{}\psi{}\ideq\psi$ ; likewise if $\psi$ is a formula then  $\odt\phi{}\psi{}\ideq\phi$ ;
\item[-] if $\odframefalse\phi\ideq\odn{\phi_1}{}{\phi_2}{}$ then
 $\vls\odframefalse\odt\phi{}\psi{}\ideq\odframefalse
\odn{\phi_1                                  }
  {}{\odframefalse\vlnos\bigg(\odt{\phi_2}
                                {}\psi    {}
                        \bigg)              }{}
$ ;
 likewise if $\psi \ideq \odn{\psi_1}{}{\psi_2}{}$ then 
$\vls\odframefalse\odt\phi{}\psi{}\ideq\odframefalse
\odn{\odframefalse\vlnos\bigg(\odt{\phi}
                                {}{\psi_1}    {}
                        \bigg)                                   }
  {}{\psi_2           }{}
$ ;
\item[-] if $\vls\phi\ideq \phi_1 \rels \phi_2$ and $\vls\psi\ideq \psi_1 \rels \psi_2$ then $\vls\odframefalse\odt\phi{}\psi{}\ideq \odt{\phi_1}{}{\psi_1}{}\rels \odt{\phi_2}{}{\psi_2}{}$ .
\end{itemize}
\end{definition}


\begin{definition}
Let $\odv{A}{\phi}{B}{\SA}$ be a derivation, and $K\vlhole$ a context. We define the derivation $K\{\phi\}$ from $K\{A\}$ to $K\{B\}$ as the derivation obtained by inserting $\phi$ in the place of the hole in $K\vlhole$.

\end{definition}

\begin{example}
If $\phi=\odn{(\fff \vlor \ttt) \as (\ttt \vlor \fff) }{\as\downarrow}{(\fff \as \ttt) \vlor (\ttt \as \fff)}{}$ and $K\vlhole= (\ttt \vlan \vlhole) \vlor (\fff \vlan \fff)$, then 
\[
K\{\phi\}= \left(\ttt \vlan \odn{(\fff \vlor \ttt) \as (\ttt \vlor \fff) }{\as\downarrow}{(\fff \as \ttt) \vlor (\ttt \as \fff)}{} \right) \vlor (\fff \vlan \fff)
\quad .
\]
\end{example}

Sometimes we will work by induction on the number of rules on a derivation. For that, it is useful to impose an order on the rules to have a notion of which one is the `last' rule of the derivation. We impose this order by sequentialising the derivation. 

\begin{definition}
 Let $\odv{A}{\phi}{B}{}$ be a derivation. We define the sequential form of $\phi$ as follows by structural induction on $\phi$:

\begin{itemize}
\item[-] if $\phi\ideq A$ is a formula, then its sequential form is given by $A$ ;
\item[-] if $ \phi \ideq \odn{\odv{A}{\phi_1}{A'}{}}{\rho}{\odv{B'}{\phi_2}{B}{}}{}$, then we consider $\phi_1$ and $\phi_2$ in sequential form: 
\[
\phi_1= \od{\odi{\odi{\odh
        {A                                          }            }
{\rho_1}{\begin{array}{c}A_2\\\vdots\\A_n\\\end{array}}{}     }
{\rho_n}{A'                                      }{}}
\qquad \mbox{and} \qquad 
\phi_2 =  \od{\odi{\odi{\odh
        {B'                                          }            }
{\rho_{n+1}}{\begin{array}{c}B_2\\\vdots\\B_m\\\end{array}}{}     }
{\rho_m}{B                                     }{}}
\]
 and the sequential form of $\phi$ is given by
\[
\odframefalse \phi=\od{\odi{\odi{\odi{\odi{\odi{\odh{A}}
{\rho_1}{\vdots}{}}
{\rho_n}{A'}{}}
{\rho}{B'}{}}{\rho_{n+1}}
{\vdots}{}}{\rho_{n+m}}
{B}{}} \quad \mbox{.}
\]
\item[-] if $\phi \ideq \odbox{\odv{A_1}{\phi_1}{B_1}{}} \rels \odbox{\odv{A_2}{\phi_2}{B_2}{}}$, then we sequentialise $\phi_1$ and $\phi_2$ to obtain
\[
\phi_1= \od{\odi{\odi{\odh
        {A_1                                          }            }
{\rho_1}{\begin{array}{c}C_2\\\vdots\\C_n\\\end{array}}{}     }
{\rho_n}{B_1                                     }{}} \qquad \mbox{and} \qquad 
\phi_2= \od{\odi{\odi{\odh
        {A_2                                          }            }
{\rho_{n+1}}{\begin{array}{c}D_2\\\vdots\\D_m\\\end{array}}{}     }
{\rho_{n+m}}{B_2                                      }{}}\]
 and the sequential form of $\phi$ is given by 
\[
\phi = \od{
\odi{\odi{\odi{\odi{\odh{A_1 \rels A_2}}{\rho_1}
{\begin{array}{c}C_2 \rels A_2\\\vdots\\C_n \rels A_2\\\end{array}}{}}{\rho_n}
{B_1 \rels A_2}{}}{\rho_{n+1}}
{\begin{array}{c}B_1 \rels D_2 \\\vdots\\B_1 \rels D_m\\\end{array}}{}}{\rho_{n+m}}
{B_1 \rels B_2}{}}\quad . \]

\end{itemize}
\end{definition}

The sequential form is not a normal form: we can choose how to sequentialise a composition by connective, by starting from either side of the connective. However we make this choice, the number of rules in the sequential form of the derivation stays nonetheless equal to the number of inference rules in its open deduction form.

\begin{example}
The sequential form of the derivation $\phi$ of example \ref{exderiv} is:
\[
\phi=
\od{\odi{\odi{\odh{(((\fff \vlor \ttt) \as (\ttt \vlor \fff)) \vlan ((\fff \bs \ttt) \vlor \ttt)) \vlan \ttt}}{\as \downarrow}
{(((\fff \as \ttt) \vlor (\ttt \as \fff)) \vlan ((\fff \bs \ttt) \vlor \ttt)) \vlan \ttt}{}}{s}
{(((\fff \as \ttt) \vlan (\fff \bs \ttt)) \vlor ((\ttt \as \fff) \vlor \ttt))\vlan \ttt}{}} \qquad.
\]

\end{example}

For some results, such as the splitting theorem in Section \ref{Spli} it is useful to consider derivations modulo certain equalities. To simplify the presentation and the case analysis, we define the Calculus of Structures presentation.  This presentation provides us with a natural way of extending an equivalence relation between formulae to an equivalence relation between derivations.

\begin{definition}\label{cos}
Let $=_{\mathcal{G}}$ be an equivalence relation on $\For$ obtained from a subset $\mathcal{G}$ of the axioms that define $=$ as per Definition \ref{eq}.

If $C =_{\mathcal{G}} C'$, there is a derivation $\odv{C}{\zeta}{C'}{}$ where $\zeta$ is composed only of equality rules corresponding to the axioms of $=_{\mathcal{G}}$. We will denote such derivations by $\odN{C}{=_{\mathcal{G}}}{C'}{} \quad$.

A derivation in sequential form
\[
\phi=
\od{\odI{\odi{\odI{\odI{\odi{\odI{\odh{A_0}}{=_{\mathcal{G}}}{A_1}{}}{\rho_1}{A_2}{}}{=_{\mathcal{G}}}{ \begin{array}{c}A_3\\\vdots\\A_n\\\end{array} }{}}{=_{\mathcal{G}}}{A_{n+1}}{}}{\rho_m}{A_{n+2}}{}}{=_{\mathcal{G}}}{A_{n+3}}{}}
\]  

has \emph{Calculus of Structures (CoS) notation} for $=_{\mathcal{G}}$ given by
\[
\phi=\od{\odi{\odi{\odh{A_0}}{\rho_1}{\begin{array}{c}A_3\\\vdots\\A_{n+1}\\\end{array}}{}}{\rho_m}{A_{n+3}}{}} \quad . \]

We define the equivalence relation $=_{\mathcal{G}}$ on derivations as $\phi_1 =_{\mathcal{G}} \phi_2$ if 
\[
\phi_1=\od{\odi{\odi{\odh{A_0}}{\rho_1}{\begin{array}{c}A_1\\\vdots\\A_n\\\end{array}}{}}{\rho_{n+1}}{A_{n+1}}{}}\quad \mbox{and} \quad 
\phi_2=\od{\odi{\odi{\odh{A'_0}}{\rho_1}{\begin{array}{c}A'_1\\\vdots\\A'_n\\\end{array}}{}}{\rho_{n+1}}{A'_{n+1}}{}}\]
 in CoS notation for $=_{\mathcal{G}}$, with $A_i =_{\mathcal{G}} A'_i$ for every $0 \leq i \leq n+1$.
\end{definition}

\begin{example}
If $=_{\mathcal{G}}$ is the equivalence relation on the set of formulae $\Forcl$ for classical logic defined by the axiom $A \vlan \ttt=A$, then
\[
\od{\odi{\odi{\odh{(((\fff \vlor \ttt) \as (\ttt \vlor \fff)) \vlan ((\fff \bs \ttt) \vlor \ttt)) \vlan \ttt}}{\as \downarrow}
{(((\fff \as \ttt) \vlor (\ttt \as \fff)) \vlan ((\fff \bs \ttt) \vlor \ttt)) \vlan \ttt}{}}{s}
{(((\fff \as \ttt) \vlan (\fff \bs \ttt)) \vlor ((\ttt \as \fff) \vlor \ttt))\vlan \ttt}{}}
\quad =_{\mathcal{G}} \quad
\od{\odi{\odi{\odh{((\fff \vlor \ttt) \as (\ttt \vlor \fff)) \vlan ((\fff \bs \ttt) \vlor \ttt) }}{\as \downarrow}
{((\fff \as \ttt) \vlor (\ttt \as \fff)) \vlan ((\fff \bs \ttt) \vlor \ttt) }{}}{s}
{((\fff \as \ttt) \vlan (\fff \bs \ttt)) \vlor ((\ttt \as \fff) \vlor \ttt)}{}} \quad .
\]
\end{example}


\begin{definition}
Let $\unit\in \Uts$ be a distinguished constant. A \emph{proof} of $A$ is a derivation $\phi$ whose premiss is $\unit$. We denote proofs by $\odr{\phi}{A}{}$.
\end{definition}

For reasons of convention, the distinguished unit for each proof system might be denoted with a different symbol, as is the case for classical logic.

\begin{example}
A proof in $\SAKS$ is a derivation with premiss $\ttt$.
\vspace{4pt}

A proof in $\SAMLLS$ is a derivation with premiss $\one$.
\vspace{4pt}

A proof in $\SABV$ is a derivation with premiss $\one$.
\end{example}

\vspace{6pt}
In the next section we will focus on showing the admissibility of certain distinguished rules.

\begin{definition}
We say that an inference rule $\rho$ is \emph{admissible} for a proof system $\SA$ if $\rho\notin\SA$ and for every proof $
\vlsmash{
\odr{}A{\SA\cup\{\rho\}}
}
$ there exists a proof $\odr{}A{\SA}$.
\end{definition}

\section{Splitting}\label{Spli}

\vllineartrue
Cut-elimination via splitting has been shown to work in a vast array of deep inference systems: linear logic \cite{Stra:02:A-Local-:ul}, multiplicative exponential linear logic \cite{Stra:03:MELL-in-:oy}, the mixed commutative/non-commutative logic BV \cite{Gugl:06:A-System:kl} and its extension with linear exponentials NEL \cite{GuglStra:02:A-Non-co:dq} and classical predicate logic \cite{Brun:06:Cut-Elim:cq}. This generality points towards the fact that the splitting procedure hinges on some fundamental properties required for cut-elimination rather than on the specifics of each system.

In particular, cut-elimination proofs via splitting are very straightforward in those systems without contractions. This suggests that it is the properties of linear rules (as opposed to contraction rules) that enable us to prove cut-elimination. Indeed, the generalisation of the splitting procedure that we show in this section allows us to fully confirm these suspicions: it is precisely because of the properties of the linear rules that we are able to prove cut-elimination for systems where they are present. In this way, we will give sufficient conditions that guarantee cut-elimination for a full class of substructural logics, similarly to \cite{Beln:82:Display-:ss, wan:98:disp, Gore:98:Substruct} where conditions for a display calculus to enjoy cut elimination are presented, or to \cite{Rasga:05:Cut-elim} where conditions for propositional based logics in the sequent calculus are presented.

 The multiplicative connectives of linear logic \cite{Gira:87:Linear-L:wm} perfectly exemplify the properties that the splitting procedure hinges on. The introduction rule for the multiplicative conjunction $\vlte$ (tensor) is given in its sequent calculus presentation as follows:

\[
\odn{\vdash A, \Phi \ \ \vdash B, \Psi}{}
{\vdash {A \vlte B}, \Phi, \Psi}{} \quad \mbox{.}
\]

Reading bottom up, we see that the multiplicative conjunction $\vlte$ requires that the context  be divided between its hypotheses. There is no communication between $\Phi$ and $\Psi$ in the proof above the tensor rule where they are united.

\[
\toks0={0.5}
\vlderivation                                    {
\vltrf\Delta  {\vdash F\{A \vlte B\}, \Gamma  } {
\vlhy         {                               } }
                                                {
\vliin\vlte {}{\vdash {A \vlte B}, \Phi,\Psi  }{
\vlhy         {\vltreeder
               {\Pi_1}{\vdash A,\Phi}{}\quad{}}}
                                               {
\vlhy         {\vltreeder
               {\Pi_2}{\vdash B,\Psi}{}\quad{}}}}
                                                {
\vlhy         {                               } }{\the\toks0}}
\]

It is precisely this multiplicative rule shape that splitting hinges on. In the sequent calculus, the presence of a main connective allows us to know exactly which rules can be applied above a cut. In deep inference, this is not possible since any rule can be applied at any depth, and we therefore focus on the behaviour of the context around a cut to tackle cut-elimination. This allows us to have a better understanding of how the cut-elimination procedure changes the proof globally. If all the connectives of a system require a splitting of the context like the multiplicative tensor does, then we can keep track of exactly how the context around a connective behaves. This allows us to split a proof into independent subproofs above every rule, just like in the example above the proof is divided into $\Pi_1$ and $\Pi_2$ above the $\vlte$ introduction rule. 
Cut-elimination is then possible following a simple idea: to show that an atomic cut involving $a$ and $\bar a$ is admissible, we follow $a$ and $\bar a$ to the top of the derivation to find two independent subderivations, the premisses of which contain the dual of $a$ and the dual of $\bar a$ respectively. In this way we obtain two proofs that don't interact above the cut, that we can rearrange to get a new cut-free proof.

\[
\newbox\boxleft\newdimen\dimenleft
\dimenleft=1.3pc     
\setbox\boxleft=\hbox{$
     \odv{{\color{blue}H_a}\vlte\odframefalse\odn{\one            }
                                               {}{{\color{blue}-a}
                                                  \vlpa
                                                  {\color{red}a}  }{}}
       {}{{\color{blue}K_a}
          \vlpa
          {\color{red}a}  }{}                                         $%
\kern-\dimenleft}
\newbox\boxright\newdimen\dimenright
\dimenright=1.3pc     
\setbox\boxright=\hbox{\kern-\dimenright$
     \odv{\odframefalse\odn{\one           }
                         {}{{\color{red}-a}
                            \vlpa
                            {\color{teal}a}}{} \vlte{\color{teal}H_{-a}}}
       {}{{\color{red}-a     }
          \vlpa
          {\color{teal}K_{-a}}}{}                                         $}
\odn{\odbox{\box\boxleft}\kern\dimenleft
 \;    \vlte   \;
     \kern\dimenright\odbox{\box\boxright}
    }
  {}{{\color{blue}K_a}
     \vlpa\odframefalse\odn{{\color{red}a}\vlte{\color{red}-a}}
                         {}{\bot                             }{}
     \vlpa{\color{teal}K_{-a}}                                   }{}
\quad\xrightarrow{\text{splitting}}\quad
\odn{{\color{blue}H_a}
     \vlte\odframefalse
     \odn{\one            }
       {}{{\color{blue}-a}
          \vlpa
          {\color{teal}a} }{}
     \vlte
     {\color{teal}H_{-a}}    }
  {}{\odv{{\color{blue}H_a\vlte-a}}
       {}{{\color{blue}K_a}       }{}
     \vlpa
     \odv{{\color{teal}a\vlte H_{-a}}}
       {}{\color{teal}K_{-a}         }{}}{}
\]

Proofs of cut-elimination by splitting therefore rely on two main properties of a proof system: the \emph{dualities} present in it to ensure that each of the independent subproofs contains the dual of an atom involved in the cut, and the \emph{shape} of the linear rules ensuring that the two proofs remain independent above the cut. It is precisely a formal characterisation of these properties that we will provide, enabling us to understand why they are enough to guarantee cut-elimination. 
\vspace{2pt}
To follow a connective through the proof from the bottom to the top, we need its scope to widen.
Accordingly, we will consider systems where the shape of the rules ensures the widening of the scope. 
In what follows, we will characterise \emph{splittable systems},  \emph{i.e.}, systems with sufficient conditions to ensure cut-elimination through a splitting procedure. 

In $\SAMLLS^{\downarrow}$ the only connective whose scope may decrease from bottom to top is $\vlpa$. The property distinguishing $\vlpa$ from $\as$ and $\vlte$ is in fact that it is the \emph{minimal} connective: it is the connective that appears in all the rules that introduce dual connectives, such as
\[
\odn{ (A \vlpa B) \vlte  (C \vlpa D)}{\vlte\downarrow}
{( A\vlte C) \vlpa ( B \vlpa D)}{}
\quad \mbox{and} \quad 
\odn{( A\vlpa B ) \as ( C\vlpa D)}{\as\downarrow}
{ (A \as C)\vlpa  (B \as D)}{}
\quad.
\] 
In every propositional system with an identity rule that introduces dualities there is such a distinguished connective $\minr$. 
In splittable systems, we will require that the scope of all connectives except for the distinguished connective $\minr$ increases from the bottom to the top of proofs, and that there be a rule $u \minr \bar u= \unit$ for every constant $u$.

Furthermore, once we have decomposed a proof into independent subproofs we will want to compose them again in such a way that we obtain a new proof. To ensure that this is possible we will require that  $\unit \rels^M \unit=\unit$ for every ${\rels}$ and that $\minr$ be associative and commutative.

\begin{definition}\label{splittable}
A system $\SA^{\downarrow}$ is \emph{splittable} if:
\begin{itemize}
\item[1.] There is a strong connective $\maxr$ with unit $\unit$ and dual $\minr$ with unit $0$,
\item[2.] $\SA$ is uniquely composed of down-rules of the form 
 \[
\odn{(A \minr B) \rels (C \minr D) }{\rels\downarrow}{(A \rels C) \minr (B \rels^m D)}{}\quad ,
\]
for every connective ${\rels} \in \Rts$.
\item[3.] There is a constant assignment $u \minr \bar u=\unit$ for every unit $u \in \Uts$,
\item[4.] $\minr$ is associative and commutative,
\item[5.] $\unit \rels^M \unit=\unit$ for every ${\rels}$.
\end{itemize}
\end{definition}

\begin{remark}
From condition 3 in Definition \ref{splittable} and the closure of $=$ under negation, $\maxr$ is associative and commutative.
\end{remark}

\begin{example}
$\SAMLLS^{\downarrow}$ is splittable, and the minimal connective $\minr$ introducing dualities is $\vlpa$.

\vllinearfalse
The linear down fragment of classical logic $\SAKS^\downarrow$ of Figure \ref{SAKSm} together with the equality rules corresponding to the axioms of example \ref{exCLeq} is splittable. The minimal connective $\minr$  introducing dualities is $\vlor$.

The down fragment $\SABVU^\downarrow$ of the system $\SABVU$ given in  Figure \ref{SABV}  together with the equality rules corresponding to the axioms of example \ref{exBVeq} is splittable. The minimal connective $\minr$  introducing dualities is $\vlpa$.

Likewise, the down fragment of $\SABV$ given in the same figure is splittable.
\end{example}

\begin{figure}
\centering
\fbox{
\begin{tabular}{c} 

\odn{ (A\vlpa B) \as  (C\vlpa D)}{\as\downarrow}
{( A\as C ) \vlpa ( B\as D )}{}\qquad \qquad\\[20pt]

\odn{ (A \vlpa B) \vlte (C\vlpa D)}{\vlte\downarrow}
{ (A\vlte C)\vlpa (B \vlpa D)}{} \qquad \qquad\\

\end{tabular}}

\caption{System $\SAMLLS^{\downarrow}$}
\label{SAMLLSm}
\end{figure}

\begin{figure}
\vllinearfalse
\centering
\fbox{
\begin{tabular}{c} 
\odn{ (A\vlor B) \as  (C\vlor D)}{\as\downarrow}
{( A\as C ) \vlor ( B\as D )}{}\qquad \qquad\\[20pt]

\odn{ (A \vlor B) \vlan (C\vlor D)}{\vlan \downarrow}
{ (A\vlan C)\vlor (B \vlor D)}{} \qquad \qquad\\
\end{tabular}}
\caption{$\SAKS^{\downarrow}$}
\label{SAKSm}
\end{figure}

The idea behind the generalisation of splitting is simple: since the scope of a connective ${\rels} \neq {\minr}$ only widens when following it from the bottom to the top of a proof, given a proof
\[
\odr{\phi}{S\{A\rels B\}}{} \quad \mbox{,}
\]
we can follow $\rels$ all the way to the top to find that $\phi$ is of the form
\[
\odn{\odr{}{A \minr Q_1}{} \rels \odr{}{B \minr Q_2}{}}{\rels\downarrow}{\odv{(A \rels B) \minr (Q_1 \rels^m Q_2)}{}{S\{A \rels B\}}{}}{} \quad \mbox{.}
\]

In other words, the proof $\phi$ splits into two subproofs that have no interaction above $\rels\downarrow$.
\vspace{2pt}

We will obtain the admissibility of certain rules as a corollary of splitting by showing that we can rearrange these independent subproofs in a suitable way. In particular, we will show that the subatomic rule that corresponds to the atomic cut rule is admissible.

\vspace{2pt}
The proof of the splitting result is done in two steps for ease of reading: shallow splitting and context reduction, just as is standard in the literature.
 As noted in \cite{Gugl:06:A-System:kl} and in  \cite{Stra:03:Linear-L:lp}, the main difficulty of splitting is finding the right induction measure for every system. In the literature, each splitting theorem for each proof system uses a different induction measure tailored specifically for it.  By providing a general splitting theorem, we not only give a formal definition of what a splitting theorem is, but also give a new one-size-fits-all induction measure that works for every splittable system, taking the search for an induction measure out of the process for designing a proof system.

\begin{lemma}\label{dual}
If $\SA^{\downarrow}$ is splittable, then for every proof 
\[ \odr{\phi}{u \minr C}{\SA^{\downarrow}}\]
 where $u\in \Uts$, there is a derivation 
\[\odv{\bar u}{\psi}{C}{\SA^{\downarrow}}\quad.\]

\end{lemma}

\begin{proof}
We take
\[
\psi \ideq
\odn{(\bar u \minr 0) \maxr \odr{\phi}{u \minr C}{}}{\maxr\downarrow}
{\odn{\bar u \maxr u}{=}{0}{} \minr 0 \minr C}{} \qquad.
\]
\end{proof}

\begin{definition}
We define $=_{\minr}$ as the equivalence relation on formulae defined by the axioms for the associativity, commutativity, unit of $\minr$ and constant assignments for $\minr$.

We define the equivalence relation $=_{\minr}$ on derivations following Definition \ref{cos}.
\end{definition}

\newcommand\size[1]{\mathopen|#1\mathclose|_{\minr}}
  \vlupdate\size
\begin{definition}
Given a derivation $\phi$, we define the \emph{length} of $\phi$ as the number of rules in $\phi$ different from the equality rules for the associativity and commutativity of $\minr$, the unit rule for $\minr$ and the unit assignments for $\minr$. We denote it by $\size{\phi}$.
 \end{definition}

It is straightforward that if $\phi =_{\minr} \psi$, then $\size{\phi}=\size{\psi}$.

\begin{theorem}[Shallow Splitting]\label{ShSpl}
If $\SA^{\downarrow}$ is splittable, for every formulae $A$, $B$, $C$,  for every connective ${\rels} \neq {\minr}$, 
for every proof 
\[
\odr{\phi}{(A \rels B) \minr C}{\SA^{\downarrow}} \]
 there exist formulae $Q_1$, $Q_2$ and derivations
 \[ 
 \odv{Q_1 \relso Q_2}{\psi}{C}{\SA^{\downarrow}} \quad \mbox{ ,} \quad
 \odr{\phi_1}{A \minr Q_1}{\SA^{\downarrow} }\quad \mbox{ and } \quad
 \odr{\phi_2}{B \minr Q_2}{\SA^{\downarrow}} \quad \mbox{,}
 \]
 with $\size{\phi_1}+\size{\phi_2} \leq \size{\phi}\qquad$.

\end{theorem}


\begin{proof}

Given a proof $\phi$ in $\SA$ of $(A \rels B) \minr C$ we reduce it to CoS notation for $=_{\minr}$. We will proceed by induction on $\size{\phi}$.

\vspace{3pt}

If $\size{\phi}=1$, then  $A=_{\minr} v, B=_{\minr} w$ and $v \rels w =_{\minr} u$, with $u \minr C =_{\minr} \unit$. By Lemma \ref{dual}, there is a derivation $\odv{\bar u}{\psi'}{C}{\SA^{\downarrow}}$ and we take:
\[
\psi \ideq 
\od{\odd{\odi{\odh{\bar v \relso \bar w}}{=}{\bar u}{}}{\psi'}{C}{}} \qquad, 
\qquad \phi_1 \ideq
 \odn{\unit}{=_{\minr} }{\odN{v}{=_{\minr}}{A}{} \minr \bar v}{} \quad \mbox{and} \quad
 \phi_2 \ideq
 \odn{\unit}{=_{\minr} }{\odN{w}{=_{\minr}}{B}{} \minr \bar w}{} \qquad.
\]

\vspace{4pt}
If $\size{\phi}=\size{\phi'}>1$, we prove the inductive step for all the possible cases of the bottom inference rule $\rho$ of $\phi$.

Inspection of the rules provides us with the following possible cases:

\begin{itemize}

\item[(1)]$\phi =_{\minr}                       \od{\odi{\odp
 {\phi'}{(A\rels B) \minr C'}{\SA^{\downarrow}}     }
{\rho}{(A\rels B) \minr C}{  }}  \qquad;$

\item[(2)]$\phi =_{\minr}                      \od{\odi{\odp
 {\phi'}{ (((A \rels B) \minr C_1)\maxr (C_2 \minr C_3))\minr C_4}{\SA^{\downarrow}}     }
{\maxr\downarrow}
{(A \rels B) \minr C_2 \minr (C_1 \maxr C_3) \minr C_4}{  }}\qquad \mbox{;}$

\item[(3)]$\phi =_{\minr}                      \od{\odi{\odp
 {\phi'}{ (((A \rels B) \minr C_1 )\relt u_{\relt}) \minr C_2}{\SA^{\downarrow}}     }
{=}
{(A \rels B) \minr C_1 \minr C_2 }{  }}\qquad \mbox{;}$

\item[(4)]$\phi =_{\minr}                      \od{\odi{\odp
 {\phi'}{ (u_{\relt} \relt ((A \rels B) \minr C_1 )) \minr C_2}{\SA^{\downarrow} }    }
{=}
{(A \rels B) \minr C_1 \minr C_2 }{  }}\qquad \mbox{;}$

\item[(5)]$\phi =_{\minr}             \od{\odi{\odp{\phi'}{(A'\rels B)\minr C}{\SA^{\downarrow}}    }
{\rho }{(A \rels B)\minr C}{    }}
\qquad;$

\item[(6)]$\phi =_{\minr}   
         \od{\odi{\odp
{\phi'}{(A\rels B') \minr C}{\SA^{\downarrow}}     }
{\rho }{(A\rels B)\minr C}{    }}
\qquad ;$

\item[(7)] $\phi =_{\minr}                       \od{\odi{\odp
 {\phi'}{ ((A \minr C_1) \rels (B \minr C_2))\minr C_3}{\SA^{\downarrow}}     }
{\rels\downarrow}
{((A \rels B) \minr (C_1  \relso C_2)) \minr C_3}{  }}\qquad \mbox{ if } \rels \mbox{ is strong ;}$

\item[(8)] $\phi =_{\minr}                     \od{\odi{\odp
 {\phi'}{ ((A \minr C_1) \relso (B \minr C_2))\minr C_3}{\SA^{\downarrow}}     }
{\relso\downarrow}
{((A \relso B) \minr (C_1  \rels C_2)) \minr C_3}{  }}\qquad \mbox{ if } \rels \mbox{ is weak ;}$

\item[(9)]$\phi =_{\minr}                       \od{\odi{\odp
 {\phi'}{ ((A \minr C_1) \rels (B \minr C_2))\minr C_3}{\SA^{\downarrow}}     }
{\rels\downarrow}
{((A \rels B) \minr (C_1  \rels C_2)) \minr C_3}{  }}\qquad \mbox{ if } \rels \mbox{ is weak ;}$

\item[(10)]$\phi =_{\minr}                       \od{\odi{\odp
 {\phi'}{(B \rels A)\minr C}{\SA^{\downarrow} }    }
{=}{(A \rels B) \minr C}{}}\qquad$ if $\rels$ is commutative ;

\item[(11)]$\phi =_{\minr}                       \od{\odi{\odp
 {\phi'}{( (A \rels B_1) \rels B_2 ) \minr C}{\SA^{\downarrow}}    }
{=}{( A \rels (B_1 \rels B_2) )\minr C}{}} \qquad$ if $\rels$ is associative ;

\item[(12)] $ \phi =_{\minr}                            \od{\odi{\odp
 {\phi'}{( A_1 \rels (A_2 \rels B) ) \minr C}{\SA^{\downarrow}}    }
{=}{( (A_1 \rels A_2) \rels B)  \minr C}{}} \qquad$ if $\rels$ is associative ;

\item[(13)] $ \phi =_{\minr}                            \od{\odi{\odp
{\phi'}{A \minr C}{\SA^{\downarrow}}}{=}
{(A \rels u_{\rels})\minr C}{}} \qquad$ if $\rels$ is unitary, with $B=_{\minr} u_{\relt}$ ;

\item[(14)] $ \phi =_{\minr}                            \od{\odi{\odp
{\phi'}{B \minr C}{\SA^{\downarrow}}}{=}
{(u_{\rels} \rels B)\minr C}{}} \qquad$ if $\rels$ is unitary,  with $A=_{\minr} u_{\relt}$ ;

\item[(15)]  $ \phi =_{\minr}                            \od{\odi{\odp
{\phi'}{u \minr C}{\SA^{\downarrow}}}{=}
{(v \rels w)\minr C}{}} \qquad$  with $A=_{\minr}v$ and $B=_{\minr} w$ .
\end{itemize}

We proceed as follows:

\begin{itemize}
\item[(1)] We can apply the induction hypothesis to $\phi'$ as $\size{\phi'}<\size{\phi}$. 

There are derivations

\[
      \psi=_{\minr}        \od{\odi{\odd{\odh
    {Q_1 \relso Q_2}                }
{\psi'}  {C'}{\SA^{\downarrow}}     }
\rho {C}     {    }}
\quad\mbox{,}\qquad
\odr{\phi_1}{A \minr Q_1}{\SA^{\downarrow}}
\qquad\mbox{and}\qquad
\odr{\phi_2}{B \minr Q_2}{\SA^{\downarrow}}
\quad\mbox{,}
\]
with $\size{\phi_1}+\size{\phi_2} \leq \size{\phi}<\size{\phi}$ .

\item[(2)] We can apply the induction hypothesis to $\phi'$ as $\size{\phi'}<\size{\phi}$.

There are derivations
\[
\vls\odframefalse
\psi=_{\minr}
\odv{ H_1 \minr H_2}{\psi'}{C_4}{\SA^{\downarrow} }            \quad \mbox{,}  \qquad
\odr{\omega_1}{ (A \rels B) \minr C_1 \minr H_1}{\SA^{\downarrow}}    \qquad
\mbox{and}\qquad
\odframefalse
\odr{\omega_2}{ C_2 \minr C_3 \minr H_2}{\SA^{\downarrow}}
      \quad \mbox{,}
\]

with $\size{\omega_1}+\size{\omega_2} \leq \size{\phi''}$.

We apply the induction hypothesis to $\omega_1$ as  $\size{\omega_1} \leq \size{\phi'}<\size{\phi}$.

There are derivations
\[
\vls\odframefalse
\odv{ Q_1 \relso Q_2}{\psi''}{C_1 \minr H_1}{\SA^{\downarrow}}             \quad \mbox{,}  \quad
\odr{\phi_1}{ A \minr Q_1}{\SA^{\downarrow}}    \quad
\mbox{,}\quad
\odframefalse
\odr{\phi_2}{  B \minr Q_2}{\SA^{\downarrow}}
      \quad \mbox{,}
\]
 with  $\size{\phi_1}+\size{\phi_2}\leq \size{\omega_1}<\size{\phi}$.

We take:

\[ \psi =_{\minr}
\odn{\odv{Q_1 \relso Q_2}
{\psi''}{C_1 \minr H_1}{}
 \maxr 
\odr{\omega_2}{ C_2 \minr C_3 \minr H_2}{}}
{\maxr\downarrow}
{(C_1 \maxr C_3) \minr C_2 \minr \odv{H_1 \minr H_2}{\psi'}{C_4}{}}{}  \quad \mbox{.}
\]

\item[(3)]We can apply the induction hypothesis to $\phi'$ as $\size{\phi'}<\size{\phi}$. There are derivations

\[
\vls\odframefalse
\odv{H_1 \relto H_2}{\psi'}{C_2}{\SA^{\downarrow}}                \quad \mbox{,}  \quad
\odr{\omega_1}{ (A \rels B) \minr C_1 \minr H_1}{\SA^{\downarrow}}   \quad
\mbox{,}\quad
\odframefalse
\odr{\omega_2}{  u_{\relt} \minr H_2}{\SA^{\downarrow}}
\quad \mbox{,}
\]

with $\size{\omega_1}+\size{\omega_2} \leq \size{\phi'}$.

By Lemma \ref{dual}, there is a derivation

\[
\vls\odframefalse
\odv{ \bar u_{\relt}}{\psi''}{H_2}{\SA^{\downarrow}}             \quad \mbox{.}
\]

\vspace{2pt}

We apply the induction hypothesis to $\omega_1$ as $\size{\omega_1}\leq \size{\phi'}<\size{\phi}$. There are derivations
\[
\vls\odframefalse
\odv{ Q_1 \relso Q_2}{\psi'''}{C_1 \minr H_1}{\SA^{\downarrow}}             \quad \mbox{,}  \quad
\odr{\phi_1}{ A \minr Q_1}{\SA^{\downarrow}}    \quad
\mbox{,}\quad
\odframefalse
\odr{\phi_2}{  B \minr Q_2}{\SA^{\downarrow}}
      \quad \mbox{,}
\]
 with  $\size{\phi_1}+\size{\phi_2}\leq \size{\omega_1}<\size{\phi}$.

We take:
\[
\psi=_{\minr}
\od{\odd{\odh{Q_1 \relso Q_2}}{\psi'''}{C_1 \minr \odv{H_1 \relto \odv{\bar u_{\relt}}{\psi''}{H_2}{}}{\psi'}{C_2}{}}{}} \qquad.
\]

\item[(4)] This case is analogous to (3).

\item[(5)] We can apply the induction hypothesis to $\phi'$ as $\size{\phi'}<\size{\phi}$. There are derivations

\[
\odv{Q_1 \rels^m Q_2}{\psi}{C}{\SA^{\downarrow}}
\quad\mbox{,}\qquad
\phi_1 \ideq \odr{\phi'_1}{\odn{A'}\rho A{} \minr Q_1}{\SA^{\downarrow}}
\qquad\mbox{and}\qquad
\odr{\phi_2}{B \minr Q_2}{\SA^{\downarrow}}
\quad\mbox{,}
\]
with $\size{\phi'_1}+\size{\phi_2} \leq \size{\phi'}$.

We have $\size{\phi_1}+\size{\phi_2}=\size{\phi'_1}+1+\size{\phi_2} \leq \size{\phi'}+1=\size{\phi}$.

\item[(6)] This case is analogous to (5).

\item[(7)] We can apply the induction hypothesis to $\phi'$ as $\size{\phi'}<\size{\phi}$. There are derivations

\[
\vls\odframefalse
\odv{ H_1 \relso H_2}{\psi'}{C_3}{\SA^{\downarrow}}             \quad \mbox{,}  \quad
 \odr{\phi_1}{ A \minr C_1 \minr H_1}{\SA^{\downarrow}}    \quad
\mbox{and}\quad
\odframefalse
 \odr{\phi_2}{  B \minr C_2  \minr H_2}{\SA^{\downarrow}}
      \quad \mbox{,}
\]

with $\size{\phi_1}+\size{\phi_2}\leq \size{\phi'}<\size{\phi}$.

We take $Q_1\ideq C_1 \minr H_1$, $Q_2\ideq C_2 \minr H_2$ and 

\[
\psi =_{\minr}
\od{\odi{
\odh{(C_1 \minr H_1) \relso (C_2 \minr H_2)}}{\relso\downarrow}
{ (C_1 \relso C_2) \minr \od{\odd{\odh{H_1 \relso H_2}}{\psi'}{C_3}{}}}{}}\qquad \mbox{.}
\]

\item[(8)] This case is analogous to (7).

\item[(9)] This case is analogous to (7).

\item[(10)] We can apply the induction hypothesis to $\phi'$ as $\size{\phi'}<\size{\phi}$. There are derivations

\[
\vls\odframefalse
\odv{ H_1 \relso H_2}{\psi'}{C}{\SA^{\downarrow}}             \quad \mbox{,}  \qquad
\odr{\omega_1}{ B \minr H_1 }{\SA^{\downarrow}}    \qquad
\mbox{and}\qquad
\odframefalse
\odr{\omega_2}{A \minr H_2}{\SA^{\downarrow}}
      \quad \mbox{,}
\]

with $\size{\omega_1}+\size{\omega_2} \leq \size{\phi'}$.

We take $Q_1\ideq H_2$, $Q_2\ideq H_1$, $\phi_1 \ideq \omega_2$, $\phi_2 \ideq \omega_1$ and

 \[
 \psi \ideq \odv{\odn{H_2 \relso H_1}{=}{H_1 \relso H_2 }{}}{\psi'}{C}{}\quad \mbox{.}
  \]

\item[(11)] We can apply the induction hypothesis to $\phi'$ as $\size{\phi'}<\size{\phi}$. There are derivations

\[
\vls\odframefalse
\odv{ H_1 \relso H_2}{\psi'}{C}{\SA^{\downarrow}}             \quad \mbox{,}  \qquad
\odr{\omega_1}{ (A \rels B_1)\minr H_1 }{\SA^{\downarrow}}    \qquad
\mbox{and}\qquad
\odframefalse
\odr{\omega_2}{  B_2 \minr H_2}{\SA^{\downarrow}}
      \quad \mbox{,}
\]

with $\size{\omega_1}+\size{\omega_2} \leq \size{\phi'}$.

We apply the induction hypothesis to $\omega_1$ as $\size{\omega_1} \leq \size{\phi'} < \size{\phi}$. There are

\[
\vls\odframefalse
\odv{ Q_1 \relso H_3}{\psi''}{H_1}{\SA^{\downarrow}}             \quad \mbox{,}  \quad
\odr{\phi_1}{  A \minr Q_1 }{\SA^{\downarrow}}    \quad
\mbox{,}\quad
\odframefalse
\odr{\omega_3}{  B_1 \minr H_3 }{\SA^{\downarrow}}
      \quad \mbox{,}
\]

with $\size{\phi_1}+ \size{\omega_3} \leq \size{\omega_1}$.

We take $Q_2\ideq H_3 \relso H_2$ and

\[
\phi_2 \ideq \odn{\odr{\omega_3}{B_1 \minr H_3}{} \rels^M
\odr{\omega_2}{B_2 \minr H_2}{}}
{\rels\downarrow}
{(B_1 \rels B_2) \minr (H_3 \relso H_2)}{} 
\quad ,\quad
\psi \ideq
\odv{\odn{Q_1 \relso (H_3 \relso H_2)}{=}
{\odv{Q_1 \relso H_3}{\psi''}{H_1}{} \relso H_2}{}}{\psi'}
{C}{} \quad \mbox{.}
\]
We have $\size{\phi_1}+\size{\phi_2} = \size{\phi_1}+\size{\omega_3}+\size{\omega_2}+1 \leq \size{\omega_1}+ \size{\omega_2}+1 \leq \size{\phi'}+1 =\size{\phi}$.

\item[(12)] This case is analogous to (11).

\item[(13)] We take $Q_1 \ideq C$, $Q_2\ideq \bar u_{\rels}$ and
\[
\psi \ideq \odn{C \relso \bar u_{\rels}}{=}{C}{} \qquad , \qquad \phi_1\ideq \odr{\phi'}{A \minr C}{} \qquad ,
\phi_2 \ideq \odn{\unit}{=_{\minr}}{\odN{u_{\rels}}{=_{\minr}}{B}{} \minr \bar u_{\rels}}{} \qquad.
\]

Then, $\size{\phi_1}+\size{\phi_2} = \size{\phi'} < \size{\phi}$.

\item[(14)] This case is analogous to (13).

\item[(15)] By Lemma \ref{dual}, there is a derivation $\odv{\bar u}{\psi'}{C}{\SA^{\downarrow}}$ and we take:
\[
\psi \ideq 
\od{\odd{\odi{\odh{\bar v \relso \bar w}}{=}{\bar u}{}}{\psi'}{C}{}} \qquad, 
\qquad \phi_1 \ideq
 \odn{\unit}{=_{\minr} }{\odN{v}{=_{\minr}}{A}{} \minr \bar v}{} \quad \mbox{and} \quad
 \phi_2 \ideq
 \odn{\unit}{=_{\minr} }{\odN{w}{=_{\minr}}{B}{} \minr \bar w}{} \qquad.
\]

\end{itemize}
\end{proof}

We can see that shallow splitting hinges precisely on the shape of the rules ${\rels}\!\downarrow$ and on the duality between constants.

\begin{remark}
The requirement for $\minr$ to be associative and commutative can be relaxed, with the condition that the rule ${\maxr} \!\downarrow$ be restricted in such a way that it corresponds to two rules
\[
\odn{(A \minr B) \maxr C}{}{(A \maxr C) \minr B}{} \qquad \mbox{and} \qquad \odn{A \maxr (B \minr C)}{}{B \minr (A \maxr C)}{} \qquad.
\]
\end{remark}

We can apply shallow splitting to the outermost connective in any context $S$, and continue applying it inductively to split any proof into independent subproofs completely. This process is formalised in the following Theorem \ref{CtxtRed}, which is a generalisation of Theorem~4.1.5 in \cite{Gugl:06:A-System:kl}.

\begin{definition}
We say that a context $H\vlhole$ is \emph{provable} if $H\{1\}=1$.
\end{definition}

\begin{theorem}[Context Reduction]\label{CtxtRed}

Let $\SA^{\downarrow}$ be a splittable system. For any formula $A$ and for any context $S\vlhole$, given a proof $\odr{\phi}{S\{A\}}{\SA^{\downarrow}}$, 
there exist a formula $K$, a provable context $H\vlhole$ and derivations
\[
\odr{\zeta}{A \minr K}{\SA^{\downarrow}} \qquad \mbox{and} \qquad 
\odv{H\{\vlhole \minr K\}}{\chi}{S\vlhole}{\SA^{\downarrow}} \qquad .
\]

\end{theorem}

\newcommand\sizerel[1] {|#1|_{ \minr}}

\begin{proof}
We proceed by induction on the number of connectives ${\rels} \neq {\minr}$ that $\vlhole$ is in the scope of in $S\vlhole$. We denote it by $\sizerel{S}$.
\vspace{2pt}

If $\sizerel{S}=0$, then $S\{A\}=_{\minr} A \minr K$ and we take $\zeta=_{\minr} \phi$ and $H\vlhole=\vlhole$.
\vspace{2pt}
If $S\{A\}=_{\minr} (S'\{A\} \relt B) \minr C$ with $\relt \neq \minr$, we apply Theorem \ref{ShSpl} to $\phi$. 
There exist derivations
 \[\odv{Q_1 \relto Q_2}{\psi}{C}{\SA^{\downarrow}} \quad \mbox{,} \quad 
 \odr{\phi_1}{ S'\{A\}\minr Q_1}{\SA^{\downarrow}} \quad \mbox{and} \quad 
 \odr{\phi_2}{ B \minr Q_2}{\SA^{\downarrow}} \quad .
 \]

 \vspace{2pt}
 
 We apply the induction hypothesis to $\phi_1$ since $\sizerel{S'}< \sizerel{S}$. 
 There are derivations
 \[
 \odr{\zeta}{A \minr K}{\SA^{\downarrow}}\quad \mbox{,} \quad 
 \odv{H'\{ \vlhole \minr K\}}{\chi'}{S'\vlhole \minr Q_1}{\SA^{\downarrow}} \quad \mbox{,}
 \]
 with $H'$ a provable context.

\vspace{2pt}
 
 We take $H\vlhole =H'\vlhole \relt^M \unit \ $. We have $H\{\unit\}=H'\{\unit\} \relt^M \unit=\unit \relt^M \unit=\unit$, and
 we can build in $\SA^{\downarrow}$
 \[
\chi \ideq
 \odn{\odv{H'\{ \vlhole \minr K\}}{\chi'}{S'\vlhole \minr Q_1}{} \relt^M \odr{\phi_2}{B \minr Q_2 }{}}{\relt\downarrow}
 {(S'\vlhole \relt B) \minr \odv{Q_1 \relto Q_2}{\psi}{C}{}}{} \quad \mbox{.}
\]

 \vspace{4pt}
 
We proceed likewise if $S\{A\}=_{\minr}(B \relt S'\{A\}) \minr C$.

\end{proof}

As a corollary of shallow splitting and context reduction we can show the admissibility of a class of up-rules. The main idea is that through splitting we can separate a proof into ``building blocks'' that are independently provable. We can then easily combine these building blocks differently to obtain a new proof with the same conclusion.

\begin{definition}
Rules of the form $\odn{(A\rels B) \maxr (C\rels^M D)}{\rels\uparrow}{(A\maxr C) \rels (B\maxr D)}{}$ are \emph{cuts}.
\end{definition}

\begin{corollary}[Admissibility of cuts]\label{adm}
Let $\SA$ be a splittable proof system.

For any formulae $A, B, C, D$, any context $S$, any connective $\rels \neq \minr$, given a proof
\[
\phi \ideq \odr{\phi'}{S\left\{\odframefalse \odn{(A\rels B) \maxr (C\rels^M D)}{\rels\uparrow}{(A\maxr C) \rels (B\maxr D)}{}\right\}}{\SA^{\downarrow} }\qquad\mbox{,}
\]
 there is a proof  
 \[
 \odr{\pi}{S\{(A\maxr C) \rels (B\maxr D)\}}{\SA^{\downarrow}  }\qquad\mbox{,}
 \]
i.e., \ cuts are admissible.

\end{corollary}

\begin{proof}
 We apply Theorem \ref{CtxtRed} to $\phi$.

There are derivations
\[
\odr{\zeta}{((A\rels B) \maxr (C\rels^M D)) \minr K}{\SA^{\downarrow}} \qquad \mbox{and} \qquad 
\odv{H\{ \vlhole \minr K \}}{\chi}{S\vlhole}{\SA^{\downarrow}}\quad \mbox{,}
\]
with $H\{\unit\}=\unit$.

We apply Theorem \ref{ShSpl} to $\zeta$. There exist derivations
 \[ 
 \odv{Q_1 \minr Q_2}{\psi}{K}{\SA^{\downarrow} } \quad \mbox{ ,} \quad
 \odr{\phi_1}{ (A\rels B) \minr Q_1}{\SA^{\downarrow}} \quad \mbox{ and } \quad
 \odr{\phi_2}{ (C \rels^M D) \minr Q_2}{\SA^{\downarrow}} \quad \mbox{.}
 \] 
 
 We apply Theorem \ref{ShSpl} to $\phi_1$ and $\phi_2$ and we obtain
  \[ 
 \odv{Q_A \relso Q_B}{\psi_1}{Q_1}{\SA^{\downarrow}} \quad \mbox{ ,} \quad
 \odr{\phi_3}{Q_A\minr A}{\SA^{\downarrow} } \quad \mbox{ and } \quad
 \odr{\phi_4}{Q_B\minr B}{\SA^{\downarrow}} \quad \mbox{,}
 \] 
 
  \[ 
 \odv{Q_C \rels^m Q_D}{\psi_2}{Q_2}{\SA^{\downarrow}} \quad \mbox{ ,} \quad
 \odr{\phi_5}{Q_C\minr C}{\SA^{\downarrow} } \quad \mbox{ and } \quad
 \odr{\phi_6}{Q_D\minr D}{\SA^{\downarrow} } \quad \mbox{.}
 \]

 We can then build the following proof in $\SA^{\downarrow}$
\[ 
\pi = 
\odv{H\left\{ \odframefalse
\odn{\odframetrue
\odn{
\odr{\phi_3}{A \minr Q_A}{} 
\maxr
 \odr{\phi_5}{C \minr Q_C}{}}
 {\maxr\downarrow}
{ (A \maxr C) \minr Q_A \minr Q_C}{}
 \rels^M 
\odn{
\odr{\phi_4}{B \minr Q_B}{} \maxr
 \odr{\phi_6}{D \minr Q_D}{}}
 {\maxr \downarrow}
{(B \maxr D) \minr Q_B \minr Q_D}{}
}
{\rels^M\downarrow}
{ \odframetrue
 ((A \maxr C) \rels (B \maxr D)) \minr
\od{\odd{\odi{\odh{(Q_A \minr Q_C) \relso (Q_B \minr Q_D)}}
{\relso\downarrow}
{\odv{Q_A \relso Q_B}{\psi_1}{Q_1}{}
\minr
\odv{Q_C \rels^m Q_D}{\psi_2}{Q_2}{}}{}}
{\psi}{K}{}}
}{}
\right\}}
{\chi}{
S \{ (A \maxr C) \rels (B \maxr D)\}}{} \qquad.
\]

\end{proof}

\begin{example}
By Corollary \ref{adm} the up fragment of $\SAMLLS$ is admissible, as is the up fragment of system $\SABVU$. 
This extends to system $\SABV$ where the units are identified. 

\vllinearfalse
Likewise, the up rules 
\[ \odn{(A \as B) \vlan (C \as D)}{\as \uparrow}{(A \vlan C) \as (B \vlan D)}{} \quad \mbox{and} \quad \odn{(A \vlor B) \vlan (C \vlan D)}{\vlor\uparrow}{(A \vlor C) \vlan (B \vlor D)}{}\]
are admissible in system $\SAKS^\downarrow$. 

\end{example}
\vspace{6pt}

The splitting procedure is therefore a very general phenomenon: it can be applied to systems with any number of connectives and units as long as certain basic equations are satisfied.

\section{Interpreting subatomic systems}\label{int}

To translate the subatomic formulae into the `usual' formulae and apply the above results to existing proof systems, we can define a simple interpretation map. The intuitive idea behind the translation is to interpret a certain assignment of units inside an atom as a positive occurrence of the atom, and the dual assignment as a negative occurrence of the atom. For example, for classical logic we interpret $\fff \as \ttt$ as a positive occurrence of the atom $a$ and $\ttt \as \fff$ as a negative one. In this way, the formula $A\ideq ((\fff \as \ttt) \vlor \ttt) \vlan (\ttt \bs \fff)$ is interpreted as $A' \ideq (a \vlor \ttt) \vlan \bar b$.

We can view the constants in the scope of an atom as a superposition of truth values. $\fff \as \ttt$ is the superposition of the two possible assignments for the atom $a$ and $\ttt \as \fff$ the superposition of the two assignments for $\bar a$. We can project onto a specific assignment by choosing which ‘side’ we read: if we read the values on the left of the atom we assign $\fff$ to
$a$ and $\ttt$ to $\bar a$ and if we read the ones on the right we assign $\ttt$ to $a$ and $\fff$ to $\bar a$.

\begin{definition}
Let $\Gor$ be the set of formulae of a propositional logic $L$. We say that the set of  subatomic formulae $\For$ is \emph{natural} for $L$ if there is a partition on the set of connectives $\Rts= \Ats \cup \Rts'$ with $\Ats \cap \Rts'=\varnothing$,such that:
\begin{itemize}
\item there is an injective map from the constants of $\Gor$ to the constants in $\Uts$;
\item there is a one to one correspondence between the connectives in $\Gor$ and the connectives in $\Rts'$;
\item there is a one to one correspondence between the set of unordered pairs of dual atoms $\{a, \bar a\}$ in $\Gor$ and the set of connectives $\Ats$.
\end{itemize}

We call the connectives in $\Ats$ \emph{atoms} as well. For each distinct pair of dual atoms we give a polarity assignment: we call one atom of the pair \emph{positive}, and the other one \emph{negative}. We will denote the atom of $\Ats$ corresponding to each pair with the same letter as the positive atom of the pair.

 We will denote the constants of $\Uts$ and the connectives in $\Rts'$ with the same symbols as their counterparts in $\Gor$.
\end{definition}

\begin{example}
The sets of subatomic formulae defined in examples \ref{exCLeq} and \ref{exMLLeq} are natural for classical logic, multiplicative linear logic and $\BV$ respectively.
\end{example}

The notion of interpretation map is easily extended to all logics for which we define a subatomic logic in the natural way. 

\begin{definition}\label{natural}

Let $\Gor$ be the set of formulae of a propositional logic $L$ with negation denoted by $\overline{\cdot}$ and equational theory denoted by $=$.
 Let $\For$ be the set of subatomic formulae with constants $\Uts$ and connectives $\Rts$ with negation denoted by  $\overline{\cdot}$ and equational theory denoted by $=$. 
A surjective partial function $I: \For \rightarrow \Gor$ is called \emph{interpretation map}. The domain of definition of $I$ is the \emph{set of interpretable formulae} and is denoted by $\For^i$.
If $F\ideq I(A)$, we say that $F$ is the \emph{interpretation} of $A$, and that $A$ is a \emph{representation} of $F$.
\vspace{2pt}

We extend the notion of interpretability to contexts: we say that $S\vlhole$ is interpretable if $S\{A\}$ is interpretable for every interpretable $A$.
\vspace{2pt}

 If $\For$ is natural for $L$, we say that an interpretation $I : \For^i \rightarrow \Gor$ is \emph{natural} when:
\begin{itemize}
\item $I(u)\ideq u$ for every constant $u$ of $\Gor$;
\item $\forall \alpha \in \Rts'$, if $A,B\in \For^i$ then $A \rels B \in \For^i$ and $ I(A \rels B)\ideq I(A) \rels I(B)$;
\item For some distinguished constants $u_1,u_2\in \Uts$, for all ${\as} \in \Ats$, $I(u_1 \as u_2)\ideq a$ and $I(u_2 \as u_1) \ideq \bar a$.
\end{itemize}

We define the \emph{natural representation $R: \Gor \rightarrow \For$ associated to $I$} for every formula $G\in \Gor$ inductively on the structure of $G$ by:
\begin{itemize}
\item $R(u)\ideq u$ if $u$ is a constant;
\item $R(a)\ideq u_1 \as u_2$  if $a$ is a positive atom;
\item $R(b) \ideq u_2 \as u_1$ if $b \ideq \bar a$ is a negative atom;
\item $R(A \rels B)\ideq R(A) \rels R(B)$ for every connective $\rels$ of $\Gor$.
\end{itemize}

For every formula $A\in \For$, $I(R(A))\ideq A$.
\end{definition}

\begin{example}\label{exCLint}
A natural interpretation for the set of subatomic formulae for classical logic defined in example \ref{exCL} is given by considering the assignments:
\[
\begin{array}{ll}
-   I(\ttt)\ideq \ttt\mbox{ ;} \quad & -   I(\fff)\ideq \fff\mbox{ ;}\\[6pt]
-  \forall {\as} \in \Ats .  I( \fff \as \fff)\ideq \fff \mbox{ ;}\quad  & -  \forall {\as} \in \Ats .  I(\ttt \as \ttt)\ideq \ttt \mbox{ ;}\\
-  \forall {\as} \in \Ats .  I( \fff \as \ttt)\ideq a \mbox{ ;}\quad  & -  \forall {\as} \in \Ats .  I(\ttt \as \fff)\ideq \bar a \mbox{ ;}\\[6pt]
-  I(A \vlor B)\ideq I(A) \vlor I(B) \mbox{ ;} \quad &-   I(A \vlan B)\ideq I(A) \vlan I(B)\mbox{ ;}\\
\end{array}
\]
where $A, B \in \For^i$, and extending it in such a way that $A \as B$ is interpretable iff $A=u$, $B=v$ with $u,v \in \{\fff, \ttt \}$ and then $I(A \as B)\ideq I(u \as v)$.

\vspace{4pt}
For example, if $A \ideq (((\fff \vlan \ttt) \as \ttt) \vlor \ttt) \vlan (\ttt \bs \fff)$, its interpretation is $I(A)= (a \vlor \ttt) \vlan \bar b$.

\vspace{4pt}

Note that the set $\For^i$ of interpretable formulae is composed by all formulae equal to a formula where an atom does not occur in the scope of another atom. Every other formula is not interpretable, such as $B \ideq ((\ttt \bs \fff) \vlan \ttt) \as \fff$.

\end{example}

\begin{example}\label{exMLLint}
A natural interpretation for the set of subatomic formulae for multiplicative additive linear logic defined in example \ref{exMLL} is given by considering the assignments:
\[
\begin{array}{ll}
-  I(\one)\ideq\one \mbox{ ;} &-   I(\bot)\ideq\bot \mbox{ ;}\\[6pt]
-  \forall {\as} \in \Ats .  I( \bot \as \bot)\ideq \bot \mbox{ ;}\quad  & -  \forall {\as} \in \Ats .  I(\one \as \one)\ideq \one \mbox{ ;}\\
-  \forall {\as} \in \Ats .  I( \bot \as \one)\ideq a \mbox{ ;}\quad  & -  \forall {\as} \in \Ats .  I(\one \as \bot)\ideq\bar a \mbox{ ;}\\[6pt]
-  I(A \vlpa B)\ideq I(A) \vlpa I(B) \mbox{ ;} \quad &-  I(A \vlte B)\ideq I(A) \vlte I(B)\mbox{ ;}\\
\end{array}
\]
where $A, B \in \For^i$, and extending it in such a way that $A \as B$ is interpretable iff $A=u$, $B=v$ with $u,v \in \{\bot, \one \}$ and then $I(A \as B)\ideq I(u \as v)$.

\vspace{4pt}

For example, for $C\ideq ((\one \vlpa \bot) \as \one)\vlte \bot$, $I(C)=a \vlte \bot$.

\vspace{4pt}

The formulae that are not interpretable are not only those equal to a formula where an atom occurs in the scope of another atom, but also those where a formula made up of units not equal to $\one$ or $\bot$ occurs in the scope of an atom, such as $(\one \vlpa \one) \as \bot$.
\vllinearfalse
\end{example}

We can easily extend the notion of interpretability to derivations.

\begin{definition}
Given an interpretation map $I$ for $\SA$, a derivation is \emph{interpretable} if every formula appearing in its sequential form is interpretable.
\end{definition}

To study proof theory through subatomic proof systems, we need to have a notion of
derivations equivalent to that of the `regular' theory. For that, we will establish a notion of
correspondence between subatomic systems and deep inference systems. In a \emph{natural}
proof system every `ordinary' proof will have a corresponding subatomic proof, and
every subatomic proof where every step has an interpretation will correspond to an
`ordinary' proof. 

In particular, in this paper we are interested in those systems where the interpretation map is built in a natural way, such that `ordinary' derivations correspond to subatomic derivations where no inference rule occurs in the scope of an atom. We call such derivations \emph{tame}.

\begin{definition}
We say that an interpretable derivation $\phi$ in $\SA$ is \emph{tame} if the only instances of rules in the scope of an atom are equality rules.
\end{definition}

\begin{example}
The derivation
\[
\odn{(\bot \vlpa \one) \as (\bot \vlpa \one)}{\as\downarrow}{(\bot \as \bot) \vlpa (\one \as \one)}{} \as \bot
\]
in $\SAMLLS$ is interpretable but is not tame.

\vspace{6pt}

The derivation 
\[
\odn{\one}{=}{(\bot \vlpa \one)}{} \as \bot
\]
is tame.
\end{example}

Note that the composition of tame derivations by any connective that is not an atom yields a tame derivation.

\begin{definition}
 Let $\SA$ be a subatomic system with a natural interpretation $I$ into the set $\Gor$ of formulae of a complete proof system $\Sy$ for a propositional logic L, with associated representation map $R$.

We say that $\SA$ is \emph{natural} for $\Sy$ when:
\begin{itemize}
\item for every tame $\SA$ derivation $\psi$ with premiss $P$ and conclusion $C$, there is a derivation $\psi'$ in $\Sy$ with premiss $I(P)$ and conclusion $I(C)$; and
\item for every derivation $\phi$ in $\Sy$ with premiss $P'$ and conclusion $C'$, there is a tame derivation $\phi'$ in $\SA$ with premiss $R(P')$ and conclusion $R(C')$.

\end{itemize}

\end{definition}

\begin{example}\label{SAKScomple}

$\SAKS^\downarrow$ of Figure \ref{SAKSm} is straightforwardly natural for the linear down fragment of classical logic of Figure \ref{SKS} given by the rules $ai\!\downarrow$ and $s$. See \cite{Aler:16:A-Study-:hc} for details.

\end{example}

\begin{example}
Likewise, system $\SAMLLS$ of Figure \ref{SAMLLS} is natural for the multiplicative fragment of system $\SLLS$ given in Figure \ref{SMLLS}.

\end{example}

\begin{example}
System $\SABV$ of Figure \ref{SABV} is correct for system $\SBV$ given in Figure \ref{SBV}. 
\end{example}

To apply the splitting result above to existing deep inference systems with a natural representation in a subatomic system, we simply need to verify that tameness is preserved by splitting.  
We can show that under simple conditions tameness is preserved by the splitting procedure, and therefore cut-free proofs obtained from tame proofs will be tame themselves. 

\begin{definition}
We say that a system $\SA$ with a natural interpretation $I$, negation $\overline{\cdot}$ and an equational theory $=$ is \emph{preservable} when:
\begin{itemize}
\item[1.] If $A$ is interpretable and $A=_{\minr} B$, then $B$ is interpretable ;
\item[2.] If $A \rels B$ is interpretable, $\rels \in \Rts$, then $A$ and $B$ are interpretable ;
\item[3.] If $A \as B$ is interpretable and $A \minr A'= 1$, $B \minr B'= 1$ then $A' \as B'$ is interpretable for ${\as} \in \Ats$ ;
\item[4.] If $A$ is interpretable, then $\overline{A}$ is interpretable ;
\item[5.] The atoms of $\Ats$ are non-commutative, non-associative and non-unitary.
\end{itemize}

\end{definition}

These conditions ensure that interpretability is preserved by duality, meaning that if an instance of a rule is interpretable, the same rule instantiated with the duals of the formulae involved is interpretable as well.

\begin{proposition}
If $SA^{\downarrow}$ is preservable and $\phi$ is tame, then the derivations $\phi_1, \phi_2$ and $\psi$ obtained from Theorem \ref{ShSpl} are tame. 
Furthermore, if $\rels$ is an atom then $\phi_1$ and $\phi_2$ are equalities.
\end{proposition}

\begin{proposition}
If $SA^{\downarrow}$ is preservable and $\phi$ is tame, then the derivation $\zeta$ obtained from Theorem \ref{CtxtRed} is tame.

Furthermore, if $\vlhole$ is not in the scope of an atom in $S\vlhole$ and $\phi$ is tame, then $\chi$ is tame.
\end{proposition}

\begin{proposition}
If $SA^{\downarrow}$ is preservable and $\phi$ is tame, and $\rels$ is not an atom, then the derivation $\pi$ is tame obtained from Corollary \ref{adm} is tame.
\end{proposition}

All the above propositions are straightforward to prove by studying each case. A detailed proof can be found in \cite{Aler:16:A-Study-:hc}.

Unsurprisingly, the only case that can pose a challenge to the preservation of tameness by the splitting procedure is when composing derivation by an atom. However, if the subatomic system under study is obtained from an `ordinary' system in a natural way, this is not problematic as can be seen from the example below.

\begin{example}
Every rule of the linear fragment of system $\KS$ for classical logic corresponds to a tame derivation in $\SAKS$. Therefore every proof in that fragment corresponds to a tame proof in $\SAKS$. 

Tameness is preserved when eliminating rule $\as\!\!\uparrow$ since every instance of a rule $\vlan\!\!\downarrow$ with the premiss equal to $\ttt$ has conclusion equal to $\ttt$ and can therefore be replaced by an equality to obtain a tame cut-free proof. Therefore, if $\rels$ is an atom and $\phi$ is tame in Corollary \ref{adm}, $\pi$ is tame as well.
\end{example}

When designing a proof system that enjoys cut-elimination, we will therefore only have to ensure that the interpretation map is preservable. This is quite an easy task, since the conditions for an interpretation map to be natural are very lenient, and therefore there is much freedom to design an interpretation to suit many needs.

\section{Conclusions}

By considering atoms as self-dual non-commutative connectives, we have classified a vast class of inference rules in a uniform and very simple way. We defined the notion of `splittable system' and we proved a general normalisation result for it, entailing cut elimination and other admissibility results as corollaries.
In papers to follow, we will extend this result to non-linear logics by providing a generalised decomposition methodology: we will characterise a wide class of systems with contraction rules in which every proof can be rewritten in such a way that it is composed of a splittable linear phase followed by a phase made-up only of contractions. The generality of this methodology will only be made possible by exploiting the regularity and granularity of the subatomic systems presented in this paper to provide generalised rewriting rules that can then be applied to a variety of systems \cite{Aler:16:A-Study-:hc}. The final result is a comprehensive theory combining decomposition and splitting and thus giving a uniform treatment for most existing logics and providing a blueprint for the design of future proof systems.

In the past, deep inference provided other unifying and simplifying mechanisms. One reason is that, in deep inference, there are many more derivations than in the sequent calculus, and derivation composition is liberal to the extreme because any semantically sound composition is allowed (this is equivalent to saying that inference can occur at any level inside a formula). Inside the larger set of derivations so obtained, it is possible to restrict derivation composition and still achieve a great expressivity and analyticity for a broad range of logics. The most prominent example of this is the formalism of nested sequents, introduced by Brünnler in \cite{Brun:07:Deep-Seq:fk} and developed by several other authors. Nested sequents allow us to systematise and give analytic calculi to modal logics that cannot be accommodated by the sequent calculus. The restriction there consists in deep inference only being allowed inside the modality prefix. An interesting question could be: if we further restrict the rules on modalities so that they have the shape that we discussed in this paper, do we still have a reasonable proof theory, and how far can we go in terms of expressiveness? The second author is working at a paper giving a positive answer to this question for an extension of BV involving a self-dual modality with similar properties to the Kleene star. In general, we think that the shape of rules introduced in this paper, when suitably restricted to unary connectives, has the potential to contribute some insights and regularity to the theory of modal logics, but we do not feel confident enough to speculate too much at this stage. All we can say is that we cannot really expect that an approach based on regularities and dualities (such as the one proposed in this paper) can work for the huge variety of modal axioms, many of which are `irregular'.

One phenomenon that we consider interesting is the following. Consider Yetter's non-commutative linear logic \cite{Yett:90:Quantale:yg}. If we follow the natural prescription of the sequent calculus, a proof system of that logic consists of a proof system for non-commutative linear logic plus the cyclic permutation rule. While non-commutative linear logic can be described by rules adhering to the subatomic shape described in this paper, cyclic permutation cannot. Therefore, seemingly, we cannot use the methodology described here for Yetter's logic. However, Di Gianantonio proved that deep inference naturally takes care of cyclicity \cite{Di-G:04:Structur:wy}, without the need to add any special rule for it. As one can easily see by inspecting the rules, Di Gianantonio's proof system for Yetter's logic is indeed splittable. In other words, this is a case in which the limitations of splittable systems are compensated by the overall more liberal derivation composition available in deep inference. This is a reason for optimism on the applicability of the techniques adopted here.

\nocite{Gugl:02:Subatomi:oh,Gugl:05:Some-New:eh}

\bibliographystyle{eptcs}
\bibliography{biblio}

\end{document}